\documentclass[12pt]{article} 
\usepackage{amsmath}
\usepackage{graphicx}
\usepackage{algorithm}
\usepackage{algpseudocode}
\usepackage{enumerate}
\usepackage{natbib}
\usepackage{url}
\usepackage{amssymb}
\usepackage{resizegather}
\usepackage{appendix}
\usepackage{graphicx}
\usepackage{float}
\usepackage{enumitem}
\usepackage{setspace}
\usepackage{rotating}
\usepackage{amsthm}
\usepackage{bm}
\usepackage{adjustbox}
\usepackage{bbm}
\usepackage{xr}
\newtheorem{proposition}{Proposition}[section]

\newtheorem{lemma}{Lemma}[section]

\newtheorem{theorem}{Theorem}[section]

\newtheorem{definition}{Definition}[section]
\newcommand{\pw}{\mathrm{pw}}
\newcommand{\X}{\mathbf{X}}
\newcommand{\U}{\mathbf{U}}

\renewcommand{\d}{\mathrm{d}}
\newcommand{\sZ}{\mathcal{Z}}
\newcommand{\TTO}{\mathrm{O}}
\newcommand{\gW}{\mathfrak{W}}
\newcommand{\gB}{\mathfrak{B}}
\newcommand{\gV}{\mathfrak{V}}
\newcommand{\FHLB}{\mathrm{FHLB}}
\newcommand{\FHUB}{\mathrm{FHUB}}
\newcommand{\PROD}{\mathrm{PROD}}

\begin{document}

\def\spacingset#1{\renewcommand{\baselinestretch}
{#1}\small\normalsize} \spacingset{1}

  \title{\bf Power-divergence copulas: A new class of Archimedean copulas, with an insurance application}
  \author{Alan R. Pearse\thanks{
    Corresponding author: alan.pearse@unimelb.edu.au} ~and Howard Bondell\hspace{.2cm}\\
    School of Mathematics and Statistics, University of Melbourne,\\
    Parkville, Victoria, Australia 3010}

    \date{7 October, 2025}
  \maketitle

\begin{abstract}
This paper demonstrates that, under a particular convention, the convex functions that characterise the phi divergences also generate Archimedean copulas in at least two dimensions. As a special case, we develop the family of Archimedean copulas associated with the important family of power divergences, which we call the \textit{power-divergence copulas}. The properties of the family are extensively studied, including the subfamilies that are absolutely continuous or have a singular component, the ordering of the family, limiting cases (i.e., the Fr\'echet-Hoeffding lower bound and Fr\'echet-Hoeffding upper bound), the Kendall's tau and tail-dependence coefficients, and cases that extend to three or more dimensions. In an illustrative application, the power-divergence copulas are used to model a Danish fire insurance dataset. It is shown that the power-divergence copulas provide an adequate fit to the bivariate distribution of two kinds of fire-related losses claimed by businesses, while several benchmarks (a suite of well known Archimedean, extreme-value, and elliptical copulas) do not. 
\end{abstract}

\noindent%

\sloppy

\spacingset{1}
\section{Introduction}
\label{sec:intro}

Copulas are flexible mathematical tools for modelling the dependence between two or more random variables. Sklar's theorem \citep{Sklar1959} states that, for any absolutely continuous random vector, the joint distribution of its elements can be uniquely decomposed into a set of marginal cumulative distribution functions (CDFs) that characterise the marginal behaviour, and a copula function that characterises the dependence structure of the random vector. A wide range of copulas exist to capture many different kinds of dependence structures in data. See \citet{Grosser2022} for a recent overview.

Archimedean copulas \citep[e.g.,][Ch. 4]{Nelsen2006} are a broad subclass of copulas characterised by a convex and strictly decreasing function called an Archimedean-copula generator. The applications of Archimedean copulas include actuarial science and insurance \citep{Kularatne2021}, engineering \citep{Orcel2021}, finance \citep{Fenech2015}, and hydrology \citep{Siamaki2024}. Archimedean copulas can also be used in time-series forecasting \citep{Patton2012} and spatial interpolation \citep{Sohrabian2021}. 

Archimedean copulas have found varied and successful applications because they are easy to construct; their properties are well understood; and they can model asymmetric dependence and dependence between extreme events (i.e., tail dependence) that are frequently observed in real-world data. This is in contrast to assuming, for example, that two or more variables in a modelling problem are jointly Gaussian or, similarly, follow a Gaussian copula model. This offers no possibility of capturing the aforementioned dependence structures, potentially leading to negative outcomes  \citep[e.g., systematic underestimation of loan-default risks in the financial sector;][]{Fenech2015}. 

Generators of Archimedean copulas may be derived from other mathematical objects. For example, \citet{Spreeuw2014} showed that Archimedean copulas can be derived from econometric utility functions. When new Archimedean-copula generators are recognised, the resulting copulas should be investigated for properties, such as lower and/or upper tail dependence, that may prove useful for modelling the diverse bivariate and multivariate dependence structures seen in real-world datasets. 

The phi ($\phi$) divergences \citep{Csiszar1963, Morimoto1963, AliSilvey1966} are statistical divergences between non-negative functions (e.g., probability densities) with myriad uses in information theory \citep{Amari2016} and statistics \citep{Pardo2006}. Recent work has established a relationship between copulas and the $\phi$ divergences \citep{Geenens2022}. While this paper also concerns (Archimedean) copulas and $\phi$ divergences, we consider a different link between them. Specifically, we demonstrate that, under a certain convention, the same convex functions that generate $\phi$ divergences are also valid Archimedean-copula generators. This observation creates an opportunity to develop new bivariate and multivariate distributions with practical utility. 

An important and well known family of $\phi$ divergences is the power divergences \citep{Cressie1984, Read1988}. They are also called the alpha ($\alpha$) divergences in information theory \citep{Amari2016}. This family, indexed by a single real-valued parameter, $\lambda \in (-\infty,\infty)$, smoothly connects several well known divergences, including the Kullback-Leibler divergence, the Pearson $\chi^2$ divergence, and the Hellinger distance. 

The power divergences represent a natural starting point for an exploration of the class of $\phi$-divergence copulas. Using the convex generator of the power divergences \citep{Cressie2002} as an Archimedean-copula generator results in the family of \textit{power-divergence} (PD) copulas. In this paper, we characterise the properties of the PD copulas. We discover that, unusually for an Archimedean copula, it is not always possible to write down the PD copulas in closed form. Nevertheless, this is not a barrier to an extensive characterisation of the properties of the PD copulas or for applied modelling. We show that the PD copulas have subfamilies that are absolutely continuous or have a singular part; that they are negatively ordered with respect to their parameter; that they are capable of modelling both negative and positive dependence; that they have the Fr\'echet-Hoeffding lower bound and upper bound as limiting cases; that they exhibit a constant moderate upper tail-dependence coefficient, while the lower tail-dependence coefficient varies with $\lambda$; and that certain values of the parameter $\lambda$ also allow the PD copulas to be valid Archimedean copulas in more than two dimensions. We also present algorithms to calculate the PD copulas and generate random variates from them. In an illustrative application, we show that PD copulas achieve an adequate fit to a dataset of Danish fire insurance claims, when several widely used Archimedean, extreme-value, and elliptical copulas fail to do so. This makes the PD copulas a useful addition to a modeller's toolbox.

The rest of the paper is structured as follows. Section \ref{sec:phi_copulas} establishes the background for our theoretical developments and introduces the phi-divergence copulas. Section \ref{sec:pd_copulas} develops the power-divergence (PD) copulas as a special case. The properties are examined in detail. Section \ref{sec:computing} considers computational aspects, such as bivariate simulation. Section \ref{sec:applied} uses the PD copulas to analyse an insurance dataset. Section \ref{sec:discussion} concludes. 

\section{Background}\label{sec:phi_copulas}

\subsection{Copulas}\label{sec:copulas}

A copula in $d$ dimensions is a $d$-variate joint cumulative distribution function (CDF) for the random vector $\U \equiv (U_1, ..., U_d)^\top$, where $U_j$ is uniformly distributed on $[0, 1]$ for $j = 1, ..., d$. We write $C(u_1, ..., u_d) \equiv \Pr(U_1 \leq u_1, ..., U_d \leq u_d), ~u_1,...,u_d \in [0, 1]$. If the copula is absolutely continuous, the \textit{copula density} is defined as $c(u_1, ..., u_d) \equiv \partial^d C(u_1, ..., u_d)/\partial u_1 \cdots \partial u_d$. If the copula is a member of a parametric family indexed by parameter $\theta$, we may write $C_\theta(u_1, ..., u_d)$ (copula) and $c_\theta(u_1, ..., u_d)$ (copula density) to emphasise this dependence on $\theta$. In this paper, our attention is primarily focused on cases where $d = 2$, but we will discuss general $d$-variate copulas for $d \geq 3$ where appropriate.

Sklar's theorem \citep{Sklar1959} makes copulas useful for modelling the joint distributions of general random vectors $\X \equiv (X_1, ..., X_d)^\top$. The joint CDF of $\X$ is $F(x_1, ..., x_d) \equiv \Pr(X_1 \leq x_1, ..., X_d \leq x_d)$, where $x_1, ..., x_d$ take values in the support of $X_1, ..., X_d$. Now, let $F_j(x_j) \equiv \Pr(X_j \leq x_j)$ and $f_j(x_j)$ denote the CDF and PDF of $X_j$, respectively, for $j = 1, ..., d$. We set $U_j \equiv F_j(X_j)$, since $F_j(X_j)$ follows a uniform distribution on $[0, 1]$ by standard properties of the CDF. Sklar's theorem states that we can obtain the joint CDF of $\X$ as,
\begin{equation}
    F(x_1, ..., x_d) = C(F_1(x_1), ..., F_d(x_d)).\label{eqn:Sklars}
\end{equation}
If $\X$ is an absolutely continuous random vector, this copula-based representation of the joint CDF is unique; otherwise, \eqref{eqn:Sklars} is non-unique. 

At this juncture, it is convenient to define three special bivariate copulas. The first is the \textit{product copula}, $C_{\PROD}(u_1, u_2) \equiv u_1 \times u_2$; this copula represents independence of the marginal variables. The other two are the Fr\'echet-Hoeffding lower bound (FHLB) and upper bound (FHUB), which, respectively, represent perfect negative and perfect positive dependence between the margins. For $u_1, u_2 \in [0,1]$, the FHLB is $C_{\FHLB}(u_1, u_2) \equiv \max\{u_1 + u_2 - 1, 0\}$, and the FHUB is $C_{\FHUB}(u_1, u_2) \equiv \min\{u_1, u_2\},~u_1,u_2\in[0,1]$. All bivariate copulas, $C(u_1, u_2)$, satisfy $C_{\FHLB}(u_1, u_2) \leq C(u_1, u_2) \leq C_{\FHUB}(u_1, u_2)$ for all $u_1, u_2 \in [0, 1]$. It is necessary to specify `bivariate copula' in the foregoing statement because the FHLB cannot be extended to $d \geq 3$ dimensions and remain a valid copula \citep[e.g., see Example 2.1 in ][]{McNeil2009}. 

\subsection{Bivariate Archimedean copulas}\label{sec:archimedean}

Archimedean copulas \citep[e.g., see][Ch. 4]{Nelsen2006} are a class of copulas characterised by a construction that proceeds from the selection of an \textit{Archimedean-copula generator} $\psi$.

\begin{definition}\label{def:generator}
Let $\psi: [0, 1] \to [0, \infty)$ be a convex and strictly decreasing function with $\psi(1) = 0$. Let $\psi^{[-1]}: [0, \infty) \to [0, 1]$ be the pseudoinverse, defined as, 
    $$\psi^{[-1]}(t) \equiv \begin{cases}
        \psi^{-1}(t)& 0 \leq t < \psi(0),\\
        0 & \psi(0) \leq t < \infty,
    \end{cases}
    $$
where $\psi^{-1}(t)$ is the inverse of $\psi(t)$ for $t\in[0, \psi(0))$, and $\psi^{-1}(0) = 1$. If $\psi(0) = \infty$, the generator $\psi$ has a \textit{strict inverse}, and $\psi^{[-1]}(t) = \psi^{-1}(t)$ for all $t \in [0, \infty)$.
\end{definition}

Constructing an Archimedean copula from a function that satisfies Definition \ref{def:generator} is straightforward. 

\begin{definition}\label{def:bivariate_archimedean}
    Assume the function $\psi$ satisfies Definition \ref{def:generator} and has pseudoinverse $\psi^{[-1]}$. Then the two-dimensional Archimedean copula \textit{generated by} $\psi$ is,
    $$
    C(u_1, u_2; \psi) \equiv \psi^{[-1]}(\psi(u_1) + \psi(u_2)).
    $$
    If $\psi$ has a strict inverse, then the associated Archimedean copula, $C(u_1, u_2; \psi) \equiv \psi^{-1}(\psi(u_1) + \psi(u_2))$, is called a \textit{strict Archimedean copula}.
\end{definition}

\noindent If the Archimedean-copula generator is indexed by a parameter, say $\theta$, then this generates a parametric family of Archimedean copulas.

\begin{definition}\label{def:bivariate_archimedean_parametric_family}
    Let $\theta \in \Theta$ be a parameter, where $\Theta$ is a parameter space. Then, if $\psi_\theta$ satisfies the conditions of an Archimedean-copula generator in Definition \ref{def:bivariate_archimedean} for all $\theta \in \Theta$, then $\{C(u_1, u_2; \psi_\theta): \theta \in \Theta\}$ is a parametric family of Archimedean copulas.  
\end{definition}

\subsection{Multivariate Archimedean copulas}\label{sec:multivariate_Archimedean}

If a function $\psi$ satisfies the requirements of Definition \ref{def:generator}, we can always construct the bivariate Archimedean copula, $C(u_1, u_2; \psi) \equiv \psi^{[-1]}(\psi(u_1) + \psi(u_2))$. The natural extension to $d \geq 3$ dimensions is to write $C(u_1, ..., u_d; \psi) \equiv \psi^{[-1]}(\psi(u_1) + \cdots + \psi(u_d))$. However, even if $C(u_1, u_2;\psi)$ is a valid copula in two dimensions, $\psi$ might not produce a valid copula for $d \geq 3$. A generator $\psi$ produces a valid $d$-dimensional Archimedean copula for all $d \geq 3$ if and only if the pseudoinverse $\psi^{[-1]}$ is \textit{completely monotone} on $t\in[0,\infty)$ \citep{Kimberling1974}; see below. 

Let $I$ be an interval of $\mathbb{R}$, and let $I^\TTO$ denote the interior of $I$. A function $h(x), ~x \in I$, is called \textit{completely monotone} on $I$ if, for all $x \in I^\TTO$, derivatives of all orders exist, and the $k$-th derivative with respect to $x$, written as $h^{(k)}(x)$, satisfies $(-1)^k \times h^{(k)}(x) \geq 0$ for $k = 0, 1, 2, ...$. The function $h(x)$ is called \textit{absolutely monotone} on $x\in I$ if $h^{(k)}(x) \geq 0$ for all $k = 0, 1, 2, ...$. The function $h(x)$ is completely monotone on $x \in I$ if and only if $h(-x)$ is absolutely monotone on $x \in I$ \citep[Def 2c]{Widder1946}. If the function $h_{a}: I_a \to (-\infty, \infty)$ is absolutely monotone on interval $I_a$ and the function $h_{b}: I_b \to I_a$ is completely monotone on the interval $I_b$, then the composite function $h_a(h_b(x))$ is completely monotone on $x \in I_b$ \citep[Thm 2b]{Widder1946}. 

Importantly for what follows, the requirement that $\psi^{[-1]}$ be completely monotone also implies that $\psi$ must have a strict inverse, and $\psi^{[-1]}(t) = \psi^{-1}(t)$ for all $t \in [0, \infty)$ \citep[pp. 151-152]{Nelsen2006}.

As a final note, complete monotonicity of $\psi^{-1}(t)$ over $t\in(0,\infty)$ is sufficient and necessary to define an Archimedean copula in all dimensions $d\geq 3$. However,  \citet{McNeil2009} show that, for a given $d \geq 3$, complete monotonicity is sufficient but not necessary. For a given $d\geq 3$, the sufficient and necessary condition to define an Archimedean copula with generator $\psi$ in $d$ dimensions is that $\psi^{[-1]}$ (the pseudoinverse, which need not be strict) is \textit{d-monotone} on $t \in [0, \infty)$. This means that $\psi^{[-1]}(t)$ has $(d-2)$ derivatives over $t\in (0, \infty)$ that satisfy $(-1)^{k}(\psi^{[-1]})^{(k)}(t) \geq 0$ for $k = 0, 1, ..., d-2$, and $(-1)^{d-2}(\psi^{[-1]})^{(d-2)}(t)$ is convex and non-increasing \citep{McNeil2009}. 

\subsection{The family of $\phi$ divergences}\label{sec:phi_divergences}

A $\phi$ divergence \citep{Csiszar1963, Morimoto1963, AliSilvey1966} is a directed, statistical divergence of one non-negative function (e.g., a potentially unnormalised probability density) from another. In the sequel, we present the construction of $\phi$ divergences needed to establish a connection with Archimedean-copula generators. 

Let $f_1$ and $f_2$ be non-negative functions (e.g., unnormalised probability densities). Below, we define the $\phi$ divergence \textit{of} $f_1$ \textit{from} $f_2$, which is written as $D_\phi(f_1 \lVert f_2)$. (In general, $D_\phi(f_1 \lVert f_2) \neq D_\phi(f_2 \lVert f_1)$, so the order of the arguments is important.) 

\begin{definition}[e.g., \citealt{Cressie2002}]\label{def:phi_divergence}
    Let $\phi: [0, \infty) \to [0, \infty)$ be a convex function that satisfies 
    \begin{enumerate}[label=(\alph*)]
        \item $\phi(1) = 0$, 
        \item $\phi'(1) = 0$,
        \item $\phi''(1) > 0$.
    \end{enumerate}  
    Then, the $\phi$ divergence \textit{of} $f_1$ \textit{from} $f_2$ is defined as,
    $$
    D_\phi(f_1 \lVert f_2) = \int f_2(s)\times \phi\left(\frac{f_1(s)}{f_2(s)}\right)~\d s,
    $$
    where, by convention, $0\times\phi(0/0) = 0$, and $0\times\phi(v/0) = v \times \{\lim_{x \to \infty} \phi(x)/x\}$.
\end{definition}

At this juncture, we note that there are other formulations of $\phi$ divergences. An alternative definition and its relationship to Definition \ref{def:phi_divergence}, are explained in Supplement S1, where we also show that it does not allow a linkage with Archimedean copulas in general. 

\subsection{Archimedean copulas via $\phi$ divergence generators}\label{sec:Archimedean_phi_copulas}

For our first result, we show that any function $\phi$ satisfying the conditions in Definition \ref{def:phi_divergence} also satisfies Definition \ref{def:generator} and hence defines a valid Archimedean-copula generator. 

\begin{proposition}\label{prop:phi_as_copula_generator}
All functions $\phi$ that satisfy (a)-(c) in Definition \ref{def:phi_divergence} are also generators of bivariate Archimedean copulas. That is, they are convex functions that satisfy $\phi(1) = 0$, and they are strictly decreasing over the interval $[0, 1]$.     
\end{proposition}

\noindent The proof and all other proofs are deferred to Supplement S7. Proposition \ref{prop:phi_as_copula_generator} then leads to the construction of bivariate $\phi$-divergence copulas.

\begin{definition}\label{def:phi_divergence_copula}
    Let $\phi$ be a function that satisfies (a)-(c) in Definition \ref{def:phi_divergence}. Let $\phi^{[-1]}$ be its pseudoinverse. Then, the associated $\phi$-divergence copula is defined as,
    $$
    C(u_1, u_2; \phi) = \phi^{[-1]}(\phi(u_1) + \phi(u_2)).
    $$
\end{definition}

The restriction to bivariate Archimedean copulas in Proposition \ref{prop:phi_as_copula_generator} is necessary because, as discussed in Section \ref{sec:multivariate_Archimedean}, stronger conditions are needed to generate Archimedean copulas in $d \geq 3$ dimensions. We can find examples of $\phi$ functions that generate Archimedean copulas in $d=2$ dimensions but not in $d \geq 3$ dimensions. A necessary (but not sufficient) condition for both complete monotonicity and $d$-monotonicity is that the pseudoinverse $\psi_\theta^{[-1]}$ be continuously differentiable over the interval $(0,\infty)$ \citep{Nelsen2006}. For a bivariate Archimedean copula, $C(u_1, u_2; \psi)$, the generator $\psi$ need not be differentiable everywhere on $(0, 1)$; it only needs to be convex on $[0, 1]$ \citep{Schweizer1983}. However, when $\psi$ is not continuously differentiable on $(0, 1)$, the pseudoinverse will also fail to be continuously differentiable over $(0,\infty)$. This implies such Archimedean-copula generators will not produce valid copulas in any given dimension $d\geq 3$ since the (sufficient and necessary) condition of $d$-monotonicity \citep{McNeil2009} requires the existence of $(d-2)$ derivatives of $\psi^{[-1]}(t)$ over all $t \in (0,\infty)$. 

Careful examination of the $\phi$ function in Definition \ref{def:phi_divergence} reveals that $\phi$ need only be twice differentiable at $x = 1$ but, according to our definition \citep[e.g.,][]{Cressie2002}, differentiability is not required for all $x \in [0, 1]$, although differentiability is also assumed in some conventions \citep[e.g., the `standard $f$ divergences' of ][pp. 54-56]{Amari2016}. One can easily construct a function that satisfies the conditions in Definition \ref{def:phi_divergence} (i.e., convex, twice differentiable at $x = 1$) but lacks a derivative at some point in $(0, 1)$. Supplement S2 gives an example. 

As a separate possibility and issue, any functions $\phi$ where $\phi(0)$ is finite, will not have completely monotone inverses, and the associated $\phi$-divergence copulas will not exist for all $d\geq3$ in general. However, this does not mean that such a $\phi$ function cannot generate a valid Archimedean copula in $d$ dimensions for a given $d \geq 3$.

\section{Power-divergence copulas}\label{sec:pd_copulas}

\subsection{Power divergences}\label{sec:CR_pd}

Power divergences \citep{Cressie1984, Read1988}, or $\alpha$ divergences \citep[e.g.,][]{Amari2016}, are an important member of the class of $\phi$ divergences. Let $\lambda \in (-\infty, \infty)$ be a power parameter. Define the $\phi$ function, 
\begin{equation}
    \phi_\lambda(x) \equiv \begin{cases}
        \frac{1}{\lambda(\lambda + 1)}\left(x^{\lambda + 1} - x + \lambda(1-x)\right)& \lambda \neq -1, 0\\
        1 - x + x\log(x) & \lambda = 0\\
        x - 1 - \log(x) & \lambda = -1,
    \end{cases}\label{eqn:CR_phi_function}
\end{equation}
which satisfies (a)-(c) in Definition \ref{def:phi_divergence}. Then, for non-negative functions $f_1$ and $f_2$, the power divergence of $f_1$ from $f_2$ is given by $D_{\phi_\lambda}(f_1 || f_2) \equiv \int f_2(s)\times\phi_\lambda(f_1(s)/f_2(s))~\d s$.

Some elementary properties of \eqref{eqn:CR_phi_function} are needed for later. First, \eqref{eqn:CR_phi_function} is strictly convex and continuously differentiable over $(0, \infty)$. Second, the first derivative with respect to $x$ is, 
\begin{equation}
    \phi_\lambda'(x) = \begin{cases}
        \lambda^{-1}(x^\lambda - 1) & \lambda \neq -1, 0,\\
    \log(x) & \lambda = 0,\\
    1- x^{-1} & \lambda = -1,
    \end{cases}\label{eqn:CR_derivative}
\end{equation}
for all $x \in (0, \infty)$. Finally, at $x = 0$, we have,
\begin{equation}
    \phi_\lambda(0) = \begin{cases}
        1/(\lambda + 1) & \lambda > -1\\
        \infty & \lambda \leq -1.
    \end{cases}\label{eqn:phi_lambda_zero}
\end{equation}

\subsection{Power-divergence (PD) copulas}\label{sec:defn_pd_copulas}

The family of power-divergence (PD) copulas can be defined by using \eqref{eqn:CR_phi_function} with Theorem \ref{prop:phi_as_copula_generator} and Definition \ref{def:bivariate_archimedean}; that is, $C_\lambda(u_1, u_2) \equiv C(u_1, u_2; \phi_\lambda),~u_1, u_2 \in [0, 1]$ for all $\lambda \in (-\infty, \infty)$. The set $\{C_\lambda(u_1, u_2): \lambda \in (-\infty, \infty)\}$ forms the family of PD copulas. The bivariate PD copulas for $\lambda \in (-\infty, \infty)$ are written as,
\begin{align}
    C_\lambda(u_1, u_2) &\equiv \phi_\lambda^{[-1]}(\phi_\lambda(u_1) + \phi_\lambda(u_2))\nonumber\\
    &= \begin{cases}
        \phi_\lambda^{-1}(\phi_\lambda(u_1) + \phi_\lambda(u_2)) & 0 \leq \phi_\lambda(u_1) + \phi_\lambda(u_2) < \phi_\lambda(0),\\
        0 & \phi_\lambda(0) \leq \phi_\lambda(u_1) + \phi_\lambda(u_2) < \infty.
    \end{cases}\label{eqn:PD_copula_defn}
\end{align}    
However, some technical challenges must be overcome to make \eqref{eqn:PD_copula_defn} practically useful. The main problem is that $\phi_\lambda^{[-1]}$ cannot be written in closed form for general $\lambda \in (-\infty, \infty)$. Letting $t \in [0, \phi_\lambda(0))$ for $\lambda \in (-\infty,\infty)$ and writing $t = \phi_\lambda(x)$, the inverse, $x = \phi_\lambda^{-1}(t)$, is found by the following analysis, which has three cases (i.e., $\lambda \neq -1, 0$, $\lambda = 0$, and $\lambda = -1$).

We begin with the two special cases, where $\lambda = -1$ and $\lambda = 0$. These have closed-form solutions, though they involve the Lambert W function \citep[e.g., see][]{Corless1996}. When $\lambda = -1$, the pseudoinverse $\phi_{-1}^{-1}(t), ~t\in[0, \infty),$ is defined by the $x \in [0, 1]$ that solves $t = x-1-\log(x)$. The (strict) inverse of $t = \phi_{-1}(x)$ is given by,
\begin{equation}
    \phi_{-1}^{-1}(t) = -\gW_0(-\exp\{-(t+1)\}),~t\in[0,\infty),\label{eqn:pseudoinverse_neg1}
\end{equation}
where $\gW_0$ denotes the principal branch of the Lambert W function. Then, the PD copula for $\lambda = -1$ can be written as,
\begin{equation}
    C_{-1}(u_1, u_2) = -\gW_0(-u_1u_2\exp\{1-(u_1 + u_2)\}),~u_1,u_2\in [0, 1].\label{eqn:copula_neg1}
\end{equation}
When $\lambda = 0$, $\phi_0(x)$ does not have a strict inverse since $\phi_0(0)=1$ is finite. The pseudoinverse $\phi_0^{[-1]}(t),~t\in [0, \infty)$, is equal to zero for $t \geq 1$, but it is defined by the value of $x \in [0, 1]$ that solves $t=1-x+x\log(x)$ for $t \in [0, 1)$. Hence, the pseudoinverse is,
\begin{equation}
\phi_{0}^{[-1]}(t) = \begin{cases}
        \exp\!\left\{\gW_{-1}((t-1)/\exp\{1\})+1\right\} & 0 \leq t < 1, \\
        0 & 1 \leq t < \infty.
    \end{cases}\label{eqn:pseudoinverse_0}
\end{equation}
where $\gW_{-1}$ is the lower branch of the Lambert W function. The PD copula for $\lambda = 0$ is,
\begin{equation}
C_0(u_1, u_2) =\begin{cases}
        \exp\!\left\{\gW_{-1}((T(u_1,u_2)-1)/\exp\{1\})+1\right\} & 0 \leq T(u_1,u_2) < 1, \\
        0 & 1 \leq T(u_1,u_2) < \infty,
    \end{cases}\label{eqn:copula_zero}
\end{equation}
where $T(u_1,u_2) = 2-u_1(1-\log(u_1)) - u_2(1-\log(u_2))$ for $u_1, u_2 \in [0,1]$.

For general $\lambda \neq -1, 0$, the inverse of $\phi_\lambda$ is strict when $\lambda < -1$; otherwise, it is not strict. For each $t \in [0, \phi_\lambda(0))$, we find $\phi^{-1}_\lambda(t)$ as the value of $x \in [0, 1]$ that solves,
\begin{equation}
    0 = x^{\lambda + 1} - (\lambda + 1)x + \lambda - \lambda(\lambda + 1)t.\label{eqn:zeros_of_pseudopolynomial}
\end{equation}
We can guarantee that \eqref{eqn:zeros_of_pseudopolynomial} always has a unique solution in the interval $[0, 1]$ (see the result below), and that $\phi_\lambda^{[-1]}(t)$ is continuous everywhere, in particular, $\lim_{t \to \phi_\lambda(0)} \phi^{[-1]}_\lambda(t) = 0$.

\begin{proposition}\label{prop:unique_solution}
    There exists a unique solution of \eqref{eqn:zeros_of_pseudopolynomial} in the interval $x \in [0, 1]$ for $\lambda \neq -1, 0$ and $t \in [0, \phi_\lambda(0))$. 
\end{proposition}

In some cases, \eqref{eqn:zeros_of_pseudopolynomial} can be solved using methods devised for root-finding in polynomials. Consider the following examples. When $\lambda$ is a positive integer, the right-hand side (RHS) of \eqref{eqn:zeros_of_pseudopolynomial} becomes a $(\lambda + 1)$-th degree polynomial in $x$, in which case the inverse is one of the roots of this polynomial. For $z \in \{2, 3, ...\}$, choosing $\lambda = (1-z)/z$ allows \eqref{eqn:zeros_of_pseudopolynomial} to be rewritten as $0 = x^{1/z} - x/z + (1-z)(1-t/z)/z$ with $t \in [0, \phi_{(1-z)/z}(0))$, where the RHS is a $z$-th degree polynomial in the variable $w \equiv x^{1/z}$. Then the solutions of \eqref{eqn:zeros_of_pseudopolynomial}, and hence also $\phi^{-1}_{(1-z)/z}$, are related to the roots of the polynomial in $w$. Similarly, with $z \in \{1, 2, 3, ...\}$, choosing $\lambda = -(z+1)$ allows \eqref{eqn:zeros_of_pseudopolynomial} to be rewritten as the $(z+1)$-th degree polynomial, $0 = zx^{z+1} - (z+1)[zt + 1]x^z+1$, where $t \in [0, \infty)$ since, in this case, $\lambda = -(z+1) < -1$. It follows that $\phi^{-1}_{-(z+1)}$ can be derived as one of the roots of this polynomial.

In a small number of cases where the RHS of \eqref{eqn:zeros_of_pseudopolynomial} is a polynomial, there is a closed-form expression for $\phi^{-1}_\lambda$ (see  Section \ref{sec:exact} and Supplement S3). However, even for the polynomial cases, closed-form solutions remain elusive in general due to the Abel-Ruffini Theorem \citep{Ruffini1813, Abel1826}, which states that polynomials of degree five and higher have no general solutions in terms of radicals (i.e., $n$-th roots for positive integer $n$). 

In general, for $\lambda \neq -1, 0$, the inverse $\phi_\lambda^{-1}(t)$ (for $t\in[0,\phi_\lambda(0))$), the pseudoinverse $\phi^{[-1]}_\lambda(t)$ (for $t\in[0,\infty)$), and the copula in \eqref{eqn:PD_copula_defn}, lack a closed form. It will be seen that this is not a barrier to characterising the properties of the PD copulas. It is also not a barrier to practical applications.

\subsection{Zero set and zero curve}\label{sec:zero_sets}

For bivariate random vector $\X\equiv(X_1, X_2)^\top$ with copula $C(u_1, u_2)$, it is possible that small values of $X_1$ and $X_2$ never occur together. In copula-based models for $\X$, this constraint can be modelled by copulas with substantial \textit{zero sets}. In some applications, this property is directly useful for modelling. For example, \citet{Konig2015} fitted Archimedean copulas with substantial zero sets to a dataset of transmission-quality metrics on a quantum network, where the zero sets were useful for capturing a prominent feature of the data. The zero sets of Archimedean copulas have also previously been related to Pareto fronts in multi-objective optimisation problems \citep{Binois2015}. The zero set and zero curve are also relevant when analysing whether Archimedean copulas are absolutely continuous or not (see Section \ref{sec:pd_copula_density}). 

For bivariate copula $C(u_1, u_2)$, $u_1,u_2 \in [0,1]$, the zero set is $\sZ(C) \equiv \{(u_1, u_2) \in [0, 1]^2: C(u_1, u_2) = 0\}$. The zero set always contains the two lines, $\{(u_1, 0), (0, u_2): u_1, u_2 \in [0, 1]\}$, but some copulas have more extensive zero sets with positive area. For example, the zero set of the FHLB, $C_{\FHLB}(u_1, u_2) = \max\{u_1 + u_2 - 1, 0\}$, is the triangular region of the unit square with the vertices $(0,0)$, $(0,1)$, and $(1,0)$. The \textit{zero curve} traces the boundary of the zero set. The zero curve of an Archimedean copula, $C(u_1, u_2; \psi)$, is the curve defined by $\psi(u_1) + \psi(u_2) = \psi(0)$ for $(u_1,u_2)\in [0,1]^2$ \citep[e.g.,][]{Nelsen2006}. The result below establishes that the zero curve and zero set of a PD copula vary with the parameter $\lambda$. 

\begin{theorem}\label{thm:zero_set}
    When $\lambda \leq -1$, the zero set of $C_\lambda(u_1, u_2)$ is $\sZ(C_\lambda)= \{(u_1, 0), (0, u_2): u_1, u_2 \in [0, 1]\}$. When $\lambda > -1$, the zero set has positive area and, as $\lambda \to \infty$, the zero set tends to the triangle with vertices $(0, 0)$, $(0, 1)$, and $(1, 0)$.
\end{theorem}

\subsection{Absolutely continuous and singular components}\label{sec:pd_copula_density}

In this section, we analyse when the PD copulas are absolutely continuous and when they contain a singular component. We begin by computing the Kendall function \citep{Genest1993} for the PD copulas. 

For any pair of absolutely continuous random variables $X_1$ and $X_2$, the bivariate random vector $(X_1, X_2)^\top$ has a Kendall function that only depends on the copula of $(X_1, X_2)^\top$. Suppose the random vector $(X_1, X_2)^\top$ has the Archimedean copula, $C(u_1, u_2; \psi_\theta)$. 
For $s \in (0, 1)$, the Kendall function is $K_\theta(s) = s - \psi_\theta(s)/\psi_\theta'(s)$, where $K_\theta(0) \equiv \lim_{s \downarrow 0} K_\theta(s)$ and $K_\theta(1) \equiv \lim_{s \uparrow 1} K_\theta(s) = 1$, and $s \downarrow 0$ and $s \uparrow 1$ denote that the one-sided limits are, respectively, approached from above and below. The Kendall function corresponds to the measure of the set, $\{(u_1, u_2) \in [0, 1]^2: C(u_1, u_2; \psi_\theta) \leq s\}$ for given $s \in [0, 1]$. The Kendall function for the PD copula is as follows: For $s \in (0, 1)$,
\begin{equation}
    K_\lambda(s) =  \begin{cases}
        \frac{\lambda}{\lambda + 1} \left( \frac{s^{\lambda + 1} - 1}{s^\lambda - 1} \right) & \lambda \neq -1, 0, \\
        \frac{s-1}{\log(s)} & \lambda = 0,\\
        \frac{s\log(s)}{s-1} & \lambda = -1.
    \end{cases} 
\end{equation}
From \citet{Genest1986}, the Archimedean copula $C(u_1, u_2; \psi_\theta)$ is absolutely continuous if $\lim_{s \downarrow 0}~\psi_\theta(s)/\psi_\theta'(s) = 0$. Otherwise, it has a singular component supported on the zero curve, with C-measure (i.e., probability mass) $K_\theta(0) = -\lim_{s \downarrow 0} \psi_\theta(s)/\psi_\theta'(s)$. As applied to the PD copulas, we obtain the following result. 

\begin{theorem}\label{thm:absolutely_continuous}
    The subfamily of power-divergence copulas, $\{C_\lambda(u_1,u_2): \lambda \leq 0\}$, is absolutely continuous. The subfamily of power-divergence copulas with $\lambda > 0$ have a singular part supported on the zero curve with C-measure $\lambda/(\lambda + 1) \in (0, 1)$.
\end{theorem}

Absolutely continuous Archimedean copulas have copula densities, $c(u_1, u_2; \psi_\theta) \equiv \partial^2 C(u_1, u_2; \psi_\theta)/\partial u_1 \partial u_2$; see \citet{Genest1986} for a standard form. Assuming $\lambda \leq 0$ (i.e., the absolutely continuous subfamily of the PD copulas), the copula density of a PD copula can be expressed as follows: For $\lambda \leq 0$ and all $(u_1,u_2)\in[0,1]^2$ that satisfy $0 \leq \phi_\lambda(u_1) + \phi_\lambda(u_2) \leq \phi_\lambda(0)$,
\begin{equation}
    c_\lambda(u_1, u_2) \equiv \frac{\partial^2 C_\lambda(u_1,u_2)}{\partial u_1\partial u_2} = \begin{cases}
        -\frac{\lambda C_\lambda(u_1, u_2)^{\lambda - 1}(u_1^\lambda - 1)(u_2^\lambda - 1)}{(C_\lambda(u_1, u_2)^\lambda - 1)^3} & \lambda \neq -1, 0,\\
        - \frac{\log(u_1)\log(u_2)}{C_0(u_1, u_2)\log(C_0(u_1,u_2))^3} & \lambda = 0,\\
        -\frac{(1-u_1^{-1})(1-u_2^{-1})}{C_{-1}(u_1,u_2)^2 (1 - C_{-1}(u_1,u_2)^{-1})} & \lambda = -1.
    \end{cases}\label{eqn:PD_copula_dens}
\end{equation}
For $\lambda \leq -1$, recall that $\phi_\lambda(0) = \infty$, and the domain of \eqref{eqn:PD_copula_dens} is the whole unit square, $[0, 1]^2$.

If the copula has a singular component (i.e., $\lambda > 0$), then the derivative \eqref{eqn:PD_copula_dens} fails to exist on the zero curve. It is still defined on a slightly restricted domain, namely $\{(u_1, u_2) \in [0,1]^2: 0 \leq \phi_\lambda(u_1) + \phi_\lambda(u_2) < \phi_\lambda(0)\}$. However, it cannot be called a density since it does not integrate to unity; in fact, it integrates to $1/(\lambda + 1)$ for any given $\lambda > 0$. 

\subsection{Ordering}\label{sec:ordering}

Assume that $\theta$ is univariate. The copula family $\{C_\theta(u_1, u_2): \theta \in \Theta\}$ may be `ordered' with respect to parameter $\theta$. The parametric family $\{C_\theta(u_1, u_2): \theta \in \Theta\}$ is called \textit{negatively ordered} if, for $\theta_1 \leq \theta_2$, $C_{\theta_1}(u_1, u_2) \geq C_{\theta_2}(u_1, u_2)$ for all $u_1, u_2 \in [0, 1]$. It is called \textit{positively ordered} if, for $\theta_1 \leq \theta_2$, $C_{\theta_1}(u_1, u_2) \leq C_{\theta_2}(u_1, u_2)$ for all $u_1, u_2 \in [0, 1]$. 
Using one of the convenient tests of ordering for families of Archimedean copulas \citep[Cor. 4.4.6]{Nelsen2006}, the next theorem establishes that PD copulas have an ordering with respect to $\lambda$.

\begin{theorem}\label{thm:ordering}
    The family of power-divergence copulas in \eqref{eqn:PD_copula_defn} is negatively ordered. That is, if $-\infty < \lambda_1 < \lambda_2 < \infty$, then $C_{\lambda_1}(u_1, u_2) \geq C_{\lambda_2}(u_1, u_2)$ for all $u_1,u_2 \in [0, 1]$.
\end{theorem}

\subsection{Limiting cases}\label{sec:limiting_cases}

The following proposition shows that the family of bivariate PD copulas includes the FHLB and FHUB as limiting cases. The family does not include the product copula.

\begin{proposition}\label{prop:FHLB_and_FHUB}
(i) The limit of $C_\lambda(u_1, u_2)$ as $\lambda \to \infty$ is the Fr\'echet-Hoeffding lower bound. (ii) The limit of $C_\lambda(u_1, u_2)$ as $\lambda \to -\infty$ is the Fr\'echet-Hoeffding upper bound.
\end{proposition}
    
\noindent The limiting cases of the PD copula are consistent with the negative ordering of the family with respect to the parameter $\lambda$, per Theorem \ref{thm:ordering}. 

\subsection{Concordance and dependence measures}\label{sec:dependence_measures}

In this section, we derive Kendall's tau \citep{Kendall1938} and the tail-dependence coefficients for the PD copulas. Kendall's tau for an Archimedean copula with generator $\psi_\theta$ is defined as $\tau(\theta) \equiv 1 + 4\int_0^1 \psi_\theta(s)/\psi_\theta'(s)~\d s$ \citep[e.g.,][Cor 5.1.4]{Nelsen2006}. Using \eqref{eqn:CR_phi_function} and \eqref{eqn:CR_derivative} in this expression gives the Kendall's tau for the PD copula: That is, for $\lambda \in (-\infty,\infty)$,
\begin{equation}
    \tau(\lambda) \equiv \begin{cases}
        1 + \frac{2}{\lambda + 1} - \frac{4\lambda}{\lambda + 1} \int_0^1\frac{s-1}{s^\lambda - 1}\d s & \lambda \neq -1, 0,\\
        3 - 4\log(2) & \lambda = 0\\
        7-2\pi^2/3 & \lambda = -1,
    \end{cases}\label{eqn:Kendalls_tau}
\end{equation}
where the integral in the first case exists for all $\lambda \neq -1, 0$, though it may not have a closed form. The last two cases correspond to numerical values of approximately $\tau(0) \simeq 0.227$ and $\tau({-1}) \simeq 0.420$. As $\lambda$ goes to $-\infty$, $\tau(\lambda)$ goes to $1$; as $\lambda$ goes to $\infty$, $\tau(\lambda)$ goes to $-1$. This follows easily from Proposition \ref{prop:FHLB_and_FHUB}. The next proposition also establishes that $\tau(\lambda)$ is a monotone function of $\lambda$.

\begin{proposition}\label{prop:Kendalls_tau_lambda}
    For $\lambda \in (-\infty, \infty)$, $\tau(\lambda)$ is a monotone decreasing function of $\lambda$.     
\end{proposition}

The expression in \eqref{eqn:Kendalls_tau} enables method-of-moments estimation for $\lambda$ based on inversion of Kendall's tau \citep{Genest1993}. Let $(X_1,X_2)^\top$ be two absolutely continuous random variables with joint CDF $F$, and let $\bar{\tau}_n$ be the sample version of Kendall's tau for the dataset, $\mathcal{D}_n=\{(x_{1i}, x_{2i})^\top \sim F: i=1,...,n\}$. The PD copula can be fitted to the dataset $\mathcal{D}_n$ by using a root-finding algorithm to compute $\hat{\lambda}$ such that $\bar{\tau}_n - \tau(\hat{\lambda}) = 0$. Proposition \ref{prop:Kendalls_tau_lambda}, together with the facts that $\lim_{\lambda \to \infty}\tau(\lambda)=-1$ and $\lim_{\lambda\to-\infty}\tau(\lambda) = 1$, implies that there is a unique value of $\hat\lambda$ associated with every value of $\bar{\tau}_n \in [-1, 1]$. This estimator is consistent and asymptotically unbiased \citep{Genest1993}. 

We now calculate the lower and upper tail-dependence coefficients of the PD copula.  For absolutely continuous random vector $\X\equiv(X_1, X_2)^\top$, the tail-dependence coefficients depend only on the copula of $\X$. Therefore, let $U_1 \equiv F_1(X_1)$ and $U_2 \equiv F_2(X_2)$, where recall that $F_1$ and $F_2$ are the marginal CDFs of $X_1$ and $X_2$. Then define $\U \equiv (U_1, U_2)^\top$, whose joint CDF is the copula $C(u_1, u_2)$ with $u_1, u_2 \in [0, 1]$. The lower-tail dependence coefficient is defined as $T_L \equiv \lim_{u \downarrow 0} \Pr(U_1 \leq u \mid U_2 \leq u)$; the upper-tail dependence coefficient is $T_U \equiv \lim_{u \uparrow 1} \Pr(U_1 \geq u \mid U_2 \geq u)$. When the copula (i.e., the CDF of $\U$) is an Archimedean copula with generator $\psi$, there are standard expressions for the tail-dependence coefficients in terms of the generator and its pseudoinverse \citep[e.g.,][]{Nelsen1997}. A copula is said to exhibit lower-tail dependence if $T_L$ is defined and $T_L > 0$; if $T_L = 0$ then there is no lower-tail dependence. The same can be said for upper-tail dependence. 

The next result shows the lower tail-dependence coefficient of the PD copula depends on $\lambda$, while the upper tail-dependence coefficient is non-zero but constant for all $\lambda \in (-\infty, \infty)$.

\begin{theorem}\label{thm:tail_dependence}
    (i) The lower tail-dependence coefficient is,
    $$
    T_L(\lambda) = \begin{cases}
    2^{1/(\lambda+1)} & \lambda < -1,\\
    0 & \lambda \geq -1.
    \end{cases}
    $$
    (ii) The upper tail-dependence coefficient is $T_U(\lambda) = 2 - \sqrt{2} \simeq 0.5858$ for all $\lambda \in (-\infty, \infty)$.
\end{theorem}

\noindent These tail-dependence coefficients are derived using the results of \citet{Charpentier2009} instead of the usual expressions in terms of the pseudoinverse of the generator originally given by \citet{Nelsen1997}; see Supplement S7.

The tail-dependence coefficients of the PD copulas suggest that the subfamily with $\lambda > -1$ could be useful for modelling phenomena with moderate upper tail dependence and an inherent restriction on the co-occurrence of small values of two random variables, $X_1$ and $X_2$. On the other hand, the subfamily of PD copulas with $\lambda \leq -1$ could be useful for modelling phenomena with moderate upper tail dependence and a wide range of lower tail-dependence behaviours. 

\subsection{PD copulas in $d \geq 3$ dimensions}\label{sec:multivariate}

It is possible to define PD copulas in $d \geq 3$ dimensions for some values of $\lambda$ but not for others. We consider three cases: $\lambda > 0$, $-1 < \lambda \leq 0$, and $\lambda \leq -1$. 

For the case where $\lambda > 0$, recall that, in order for $\phi_{\lambda}^{[-1]}(t)$ to be $3$-monotone over $t \in [0, \infty)$, it is first necessary that the derivative $(\phi_{\lambda}^{[-1]})'(t)$ exists everywhere on $t \in (0, \infty)$ \citep{McNeil2009}. The next theorem shows it does not. In particular, it fails to exist at $t = 1/(\lambda + 1)$. See Supplement S7 for the details.

\begin{lemma}\label{lem:no_derivative_at_phi0}
    Let $\lambda > 0$. Then the derivative of $\phi^{[-1]}_\lambda(t)$, $t \in [0, \infty)$, does not exist at $t = \phi_\lambda(0) = 1/(\lambda + 1)$.
\end{lemma}

\noindent This immediately implies the following. 

\begin{theorem}\label{thm:no_pd_copulas_d_geq_3_lambda_above_0}
    For $\lambda > 0$ and $u_1, ..., u_d \in [0,1]$, $\phi^{[-1]}_\lambda(\phi_\lambda(u_1) + \cdots + \phi_\lambda(u_d))$ is not a valid copula for any $d \geq 3$.  
\end{theorem}

The situation is different (and more complicated) when $-1 < \lambda \leq 0$. Like the $\lambda > 0$ case, the inverse is non-strict, and $\phi_\lambda(0) = 1/(\lambda + 1)$. However, unlike the previous case, the following result guarantees that the derivative $(\phi^{[-1]}_{\lambda})'(t)$ exists for all $t \in [0, \infty)$, even at $t = 1/(\lambda + 1)$, and that $-(\phi^{[-1]}_{\lambda})'(t) \geq 0$ for all $t \in (0, \infty)$. 

\begin{lemma}\label{lem:derivative_exists}
    For $-1 < \lambda \leq 0$, the pseudoinverse $\phi_\lambda^{[-1]}(t)$ is differentiable over all $t \in (0, \infty)$, and the derivative $(\phi_\lambda^{[-1]})'(1/(\lambda + 1)) = 0$. Further, $-(\phi^{[-1]}_{\lambda})'(t) \geq 0$.
\end{lemma}

\noindent Yet $\phi^{[-1]}_{\lambda}(t)$ is not necessarily $3$-monotone for all $-1 < \lambda \leq 0$. Per \citet{McNeil2009}, $-(\phi^{[-1]}_{\lambda})'(t)$ must also be convex over $t \in [0, \infty)$ in order for $\phi_\lambda^{[-1]}(t)$ to be $3$-monotone. The next result demonstrates that $-(\phi^{[-1]}_\lambda)'(t)$ fails to be convex on $t \in [0, \infty)$ unless $\lambda \leq -0.5$.

\begin{lemma}\label{lem:fails_to_be_convex}
    For $-1 < \lambda \leq 0$, the function $-(\phi^{[-1]}_\lambda)'(t)$ is non-increasing on $t \in [0, \infty)$. However, while the  function $-(\phi^{[-1]}_\lambda)'(t)$ is convex on $t \in [0, \infty)$ for $-1 < \lambda \leq -0.5$, it fails to be convex for some values of $t$ when $-0.5 < \lambda \leq 0$.
\end{lemma}

\noindent As a corollary, all pseudoinverses in the set $\{\phi_{\lambda}^{-1}: -1 < \lambda \leq -0.5\}$ are (at least) $3$-monotone on $t \in [0, \infty)$. Then we obtain the following result.

\begin{theorem}\label{thm:three_monotone}
    For $-1 < \lambda \leq -0.5$ and $u_1, u_2, u_3 \in [0,1]$, $C_\lambda(u_1, u_2, u_3) \equiv \phi^{[-1]}_\lambda(\phi_\lambda(u_1) + \phi_\lambda(u_2) + \phi_\lambda(u_3))$ is a valid copula in three dimensions. 
\end{theorem}

\noindent On the other hand, PD copulas with $\lambda > -0.5$ are restricted to two dimensions only. 

For $\lambda \leq -1$, the inverses $\{\phi_\lambda^{-1}: \lambda \leq -1\}$ are strict and, in fact, the inverses have derivatives of all orders on $t \in (0,\infty)$. This follows from the Inverse Function Theorem since, by inspection, $\phi_{\lambda}(x)$ in \eqref{eqn:CR_phi_function} is infinitely differentiable on $x \in (0, \infty)$ when $\lambda \leq -1$. The only question is whether the derivatives of the inverse satisfy $(-1)^k(\phi_{\lambda}^{-1})^{(k)}(t) \geq 0$ for all $k = 1, 2, ...$, or the weaker condition of $d$-monotonicity for some $d \geq 3$. Regarding $d$-monotonicity, Lemma \ref{lem:fails_to_be_convex} and its proof in Supplement S7 can be extended straightforwardly to $\lambda < -1$, so $\phi_{\lambda}^{-1}(t)$ is at least $3$-monotone on $t \in [0,\infty)$. Although Lemma \ref{lem:fails_to_be_convex} does not cover $\lambda = -1$, we show below that a stronger result holds for this case anyway. 

As for complete monotonicity, a general result for all inverses with $\lambda \leq -1$ remains elusive. However, the result below shows that $\phi_{-1}^{-1}(t)$ and $\phi_{-2}^{-1}(t)$ are completely monotone on $t \in [0, \infty)$. Part (i) of Lemma \ref{lem:pd_copulas_completely_monotone} follows from properties of the principal branch of the Lambert W function \citep[due to, e.g.,][]{Kalugin2012}. Our proof of part (ii)  uses Fa\`{a} di Bruno's formula and a property of incomplete exponential Bell polynomials \citep{Bell1934, Comtet1974}. See  Supplement S6 and S7 for details.
    
\begin{lemma}\label{lem:pd_copulas_completely_monotone}
    (i) Let $\lambda = -1$. The strict inverse $\phi^{-1}_{-1}(t)$ in \eqref{eqn:pseudoinverse_neg1} is completely monotone on $t\in [0, \infty)$. (ii) Let $\lambda = -2$. The strict inverse $\phi^{-1}_{-2}(t)$ is completely monotone on $t\in [0, \infty)$.
\end{lemma}

\noindent It immediately follows from Lemma \ref{lem:pd_copulas_completely_monotone} that the PD copulas with $\lambda = -1$ and $\lambda = -2$ are valid copulas in all dimensions $d \geq 3$; see the next result.

\begin{theorem}\label{thm:completely_monotone_cases}
    For all $d \geq 3$ and $u_1, ..., u_d \in [0,1]$, $C_{-1}(u_1, ..., u_d) \equiv \phi_{-1}^{-1}(\phi_{-1}(u_1) + \cdots + \phi_{-1}(u_d))$ and $C_{-2}(u_1, ..., u_d) \equiv \phi_{-2}^{-1}(\phi_{-2}(u_1) + \cdots + \phi_{-2}(u_d))$ are valid copulas.
\end{theorem}

We conjecture that all members of $\{\phi_{\lambda}^{-1}: \lambda \leq -1\}$ are completely monotone on $t \in [0, \infty)$. However, in the absence of a general proof, it is prudent to check that $\phi_\lambda^{-1}(t)$ is $d$-monotone on $t\in[0,\infty)$ for given $\lambda \leq -1$ and $d\geq3$ if such a PD copula is required. Using, for example, Mathematica \citep{Mathematica}, it is straightforward to check whether $\phi^{-1}_{\lambda}(t)$ is $d$-monotone on $t\in [0,\infty)$ for given $d\geq3$ and $\lambda < -1$. Letting $\gamma > 1$ and setting $\lambda = -\gamma$, simply note that $(\phi_{-\gamma}^{-1})'(t) = -\gamma\phi^{-1}_{-\gamma}(t)^\gamma/(1-\phi^{-1}_{-\gamma}(t)^\gamma)$ for all $t\in (0,\infty)$ by the Inverse Function Theorem, and use this fact to compute (by symbolic differentiation) the $(d-2)$-th, $(d-1)$-th, and $d$-th derivatives of $\phi^{-1}_{-\gamma}(t)$ to verify that $(-1)^{d-2}(\phi_\lambda^{-1})^{(d-2)}(t) \geq 0$ and that $(-1)^{d-2}(\phi_\lambda^{-1})^{(d-2)}(t)$ is non-increasing and convex on $t \in (0, \infty)$. We need only check these three derivatives (and no preceding derivatives) because Proposition 2.3 of \citet{McNeil2009} shows it suffices to only check that $(-1)^{d-2}(\phi_\lambda^{-1})^{(d-2)}(t)$ is non-negative, convex, and non-increasing for a given $d \geq 3$.

\section{Computational aspects}\label{sec:computing}

\subsection{Exact formulas for certain values of $\lambda$}\label{sec:exact}

Though Proposition \ref{prop:unique_solution} guarantees that the pseudoinverse $\phi_\lambda^{[-1]}(t)$ exists for $t \in [0, \infty)$, the PD copula has no closed form for general $\lambda \neq 0, -1$. A non-exhaustive list of exceptions is $\lambda \in \{-4, -3, -2, -2/3, -1/2, 1, 2, 3\}$. In these cases, \eqref{eqn:zeros_of_pseudopolynomial} or a suitable rearrangement reduces to the problem of finding the zeros of a quadratic, cubic, or quartic equation in $x$ or $w\equiv x^{\lambda+1}$. Such problems have been studied for several centuries. Regardless, deriving the exact solutions can involve considerable effort, even with the aid of a computer algebra system. Therefore, we do not examine every possibility. Supplement S3 presents three illustrative examples, where $\lambda \in \{-2, -0.5, 1\}$.

\subsection{Computing the PD copula and its density}\label{sec:compute_copula}

For $\lambda \neq -1, 0$, the inverse $\phi_\lambda^{-1}(t)$ is given by the solutions of \eqref{eqn:zeros_of_pseudopolynomial}. Recall that Proposition \ref{prop:unique_solution} guarantees an appropriate solution in the range $\phi_\lambda^{-1}(t) \in [0, 1]$ exists and is unique for all $t \in [0, \phi_\lambda(0))$. Algorithm \ref{alg:compute_copula} uses this insight to develop a numerical routine for computing the PD copula for any $u_1, u_2 \in [0, 1]$. Copulas for $\lambda \in \{-\sqrt{2}, \sqrt{2}\}$ (i.e., values of $\lambda$ that do not admit closed-form representations of the corresponding PD copulas) are plotted in Fig. \ref{fig:sim_copulas} to illustrate this numerical routine.

\begin{algorithm}
\caption{Computing the power-divergence copula for general $\lambda \in (-\infty, \infty)$.}\label{alg:compute_copula}
\begin{algorithmic}
\State Let $\lambda \in (-\infty, \infty)$.
\State Let $u_1, u_2 \in [0,1]$.
\State Compute $t = \phi_\lambda(u_1) + \phi_\lambda(u_2)$ using \eqref{eqn:CR_phi_function}.
\If{$t \leq \phi_\lambda(0)$, which is always true if $\lambda \leq -1$,}
    \If{$\lambda \neq 0, -1$,}
    \State Compute the sole solution of \eqref{eqn:zeros_of_pseudopolynomial} in $[0, 1]$ by a root-finding algorithm; call it $x^*$. 
    \State $C_\lambda(u_1, u_2) = x^*$.
    \Else
    \State Compute $C_{-1}(u_1, u_2)$ with \eqref{eqn:copula_neg1} or $C_{0}(u_1, u_2)$ with \eqref{eqn:copula_zero} as appropriate.
    \EndIf
\Else
    \State $C_\lambda(u_1, u_2) = 0$.
\EndIf
\end{algorithmic}
\end{algorithm}

\begin{figure}[!ht]
    \centering
    \includegraphics[width=0.8\textwidth]{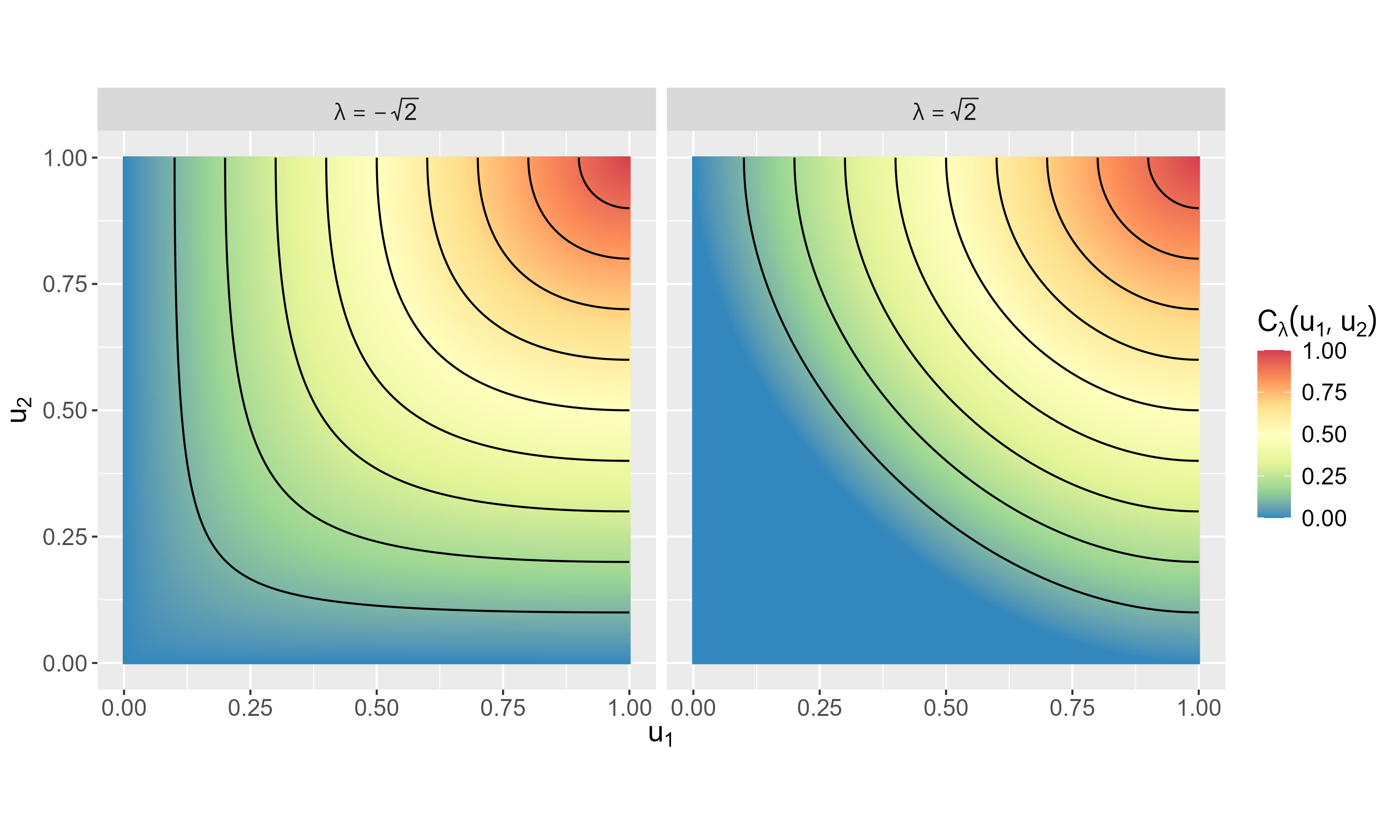}
    \caption{Plots of $C_\lambda(u_1, u_2)$ for $\lambda \in \{-\sqrt{2},\sqrt{2}\}$ and $u_1, u_2 \in [0, 1]$ obtained numerically.}
    \label{fig:sim_copulas}
\end{figure}

Computing the copula density (or at least the mixed partial derivative in \eqref{eqn:PD_copula_dens}) on its domain of definition is straightforward. For any particular $u_1, u_2 \in [0, 1]$ and $\lambda \in (-\infty,\infty)$, check if $\phi_\lambda(u_1) + \phi_\lambda(u_2) < \phi_\lambda(0)$. If so, simply compute the copula $C_\lambda(u_1, u_2)$ by Algorithm \ref{alg:compute_copula}, and then substitute this value into \eqref{eqn:PD_copula_dens}. For $\lambda \geq 0$, recall that the `density' is only for the absolutely continuous part of the copula. The singular component of the copula is supported on the zero curve and, when $\lambda > 0$, has $100\times \lambda/(\lambda + 1)\%$ of the probability mass of the copula. The remaining probability mass is distributed according to the `density' in the absolutely continuous part. Fig. \ref{fig:sim_copula_densities} shows the computed values of $\partial^2C_\lambda(u_1, u_2)/\partial u_1\partial u_2$ for $\lambda = -\sqrt{2}$ and $\lambda = \sqrt{2}$ over the relevant domains of definition.

\begin{figure}[!ht]
    \centering
    \includegraphics[width=\textwidth]{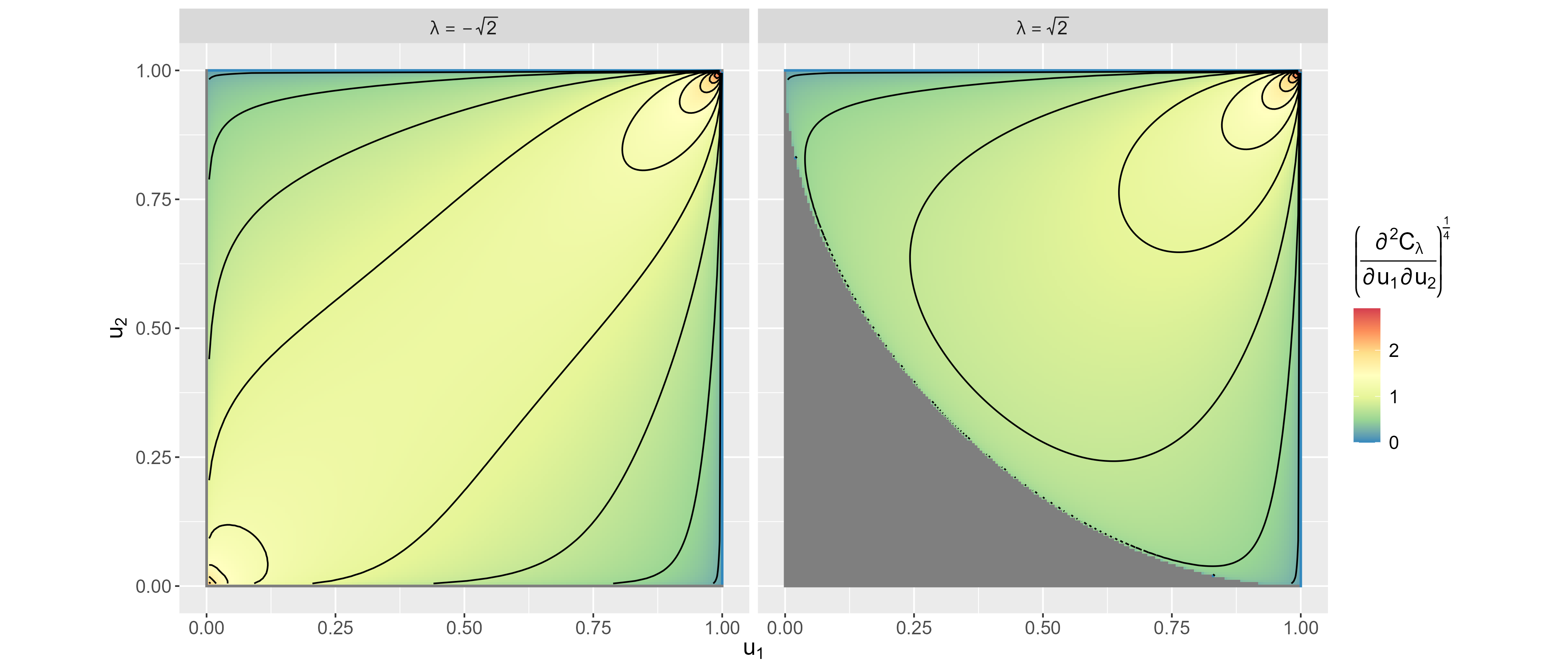}
    \caption{For $\lambda = -\sqrt{2}$ (left panel) and $\lambda = \sqrt{2}$ (right panel), plots of the fourth root of $\partial^2C_\lambda(u_1, u_2)/\partial u_1\partial u_2$ for $u_1, u_2 \in [0, 1]$. The fourth root was used for display purposes only.}
    \label{fig:sim_copula_densities}
\end{figure}

\subsection{Simulating from the PD copula}\label{sec:simulate_from_copula}

Simulation from the bivariate PD copula is enabled by the well known `conditional distribution method' \citep[e.g.,][Sec 2.9]{Nelsen2006}. Supplement S4 gives the details of the algorithm for the PD copulas. Fig. \ref{fig:simulated_data} plots data simulated from the PD copulas for selected values of $\lambda$, namely $\lambda \in \{-10, -3, -2, -0.5, 1, 2, 3, 10\}$. Some notable features of the plots include the upper tail dependence apparent in the simulated random variates and the heavy zero curves when $\lambda = 1, 2, 3, 10$. This reflects the large C-measures of the zero curves, which go up to approximately $0.91$ when $\lambda = 10$. For $\lambda = -0.5$, a zero set is apparent but the C-measure of the zero curve is 0, meaning no random variates sit on the zero curve. 

\begin{figure}[!ht]
    \centering
    \includegraphics[width=\linewidth]{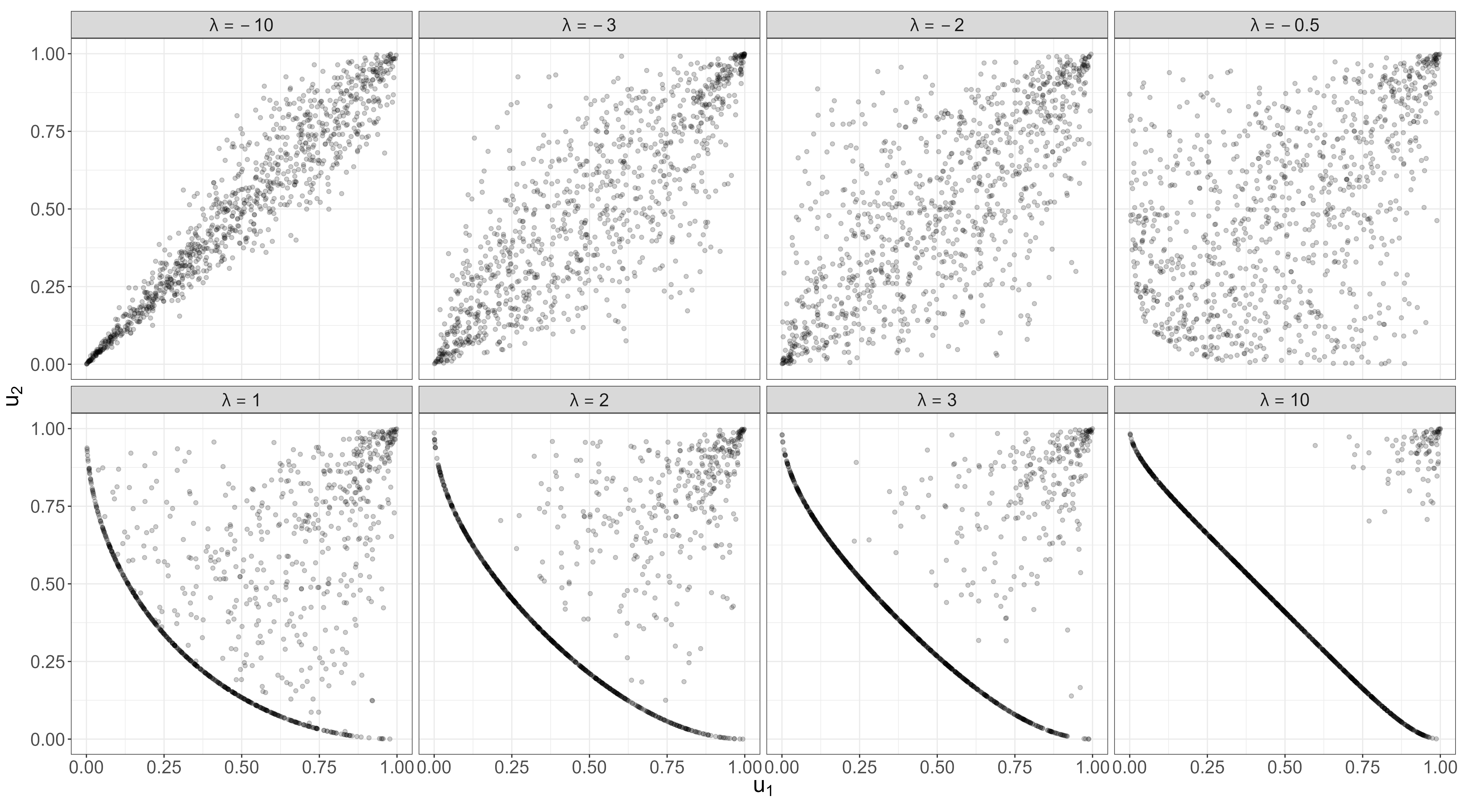}
    \caption{Scatterplots of 1000 bivariate realisations from the power-divergence copulas, with $\lambda \in \{-10, -3, -2, -0.5, 1, 2, 3, 10\}$.}
    \label{fig:simulated_data}
\end{figure}

\section{Analysis of Danish fire insurance data}\label{sec:applied}

In this section, we analyse a dataset of Danish insurance claims for businesses experiencing fire-related losses between 1980 and 1990. The dataset comprises 2,167 records of the losses to building, contents, and profits in millions of Danish krone (MDK), adjusted for inflation to the year 1985. The data are available in the R package `fitdistrplus' \citep{fitdistrplus}. We refer to these data as the `Danish fire insurance data'. 

An insurance company or their reinsurers may be interested to know, for example, the joint exceedance probabilities of losses to buildings, contents, and profits to inform pricing of policies, etc. Therefore, characterising the (potentially complicated) dependence between these variables is of interest, and copulas are uniquely suited to this task. In what follows, we demonstrate the PD copulas achieve an adequate fit to the Danish fire insurance data, whereas a suite of well known Archimedean, extreme-value, and elliptical copulas do not. 

The Danish fire insurance data were previously analysed by \citet{Haug2011} and \citet{Kularatne2021}, who examined the bivariate relationship between losses to building and losses to contents for small subsets of the data. We did not follow their choices in processing and modelling the data. Instead, our analysis examined the bivariate relationship between the sum of losses to building and contents, designated as `material losses', and the loss to profits. Records with no loss to profits were removed, so the following inferences are conditional on non-zero loss to profits. Fig. \ref{fig:danish_data} shows the remaining $616$ records. The losses are shown in MDK, but they are also shown following a transformation into `copula data' using a non-parametric rank-based transformation of the margins. The copula data reveal the dependence structure of the material losses and losses to profits, which appears to have moderate upper tail dependence. The plot also suggests that small values of the marginal variables do not occur together. This is likely because the dataset only contains records of insurance claims totalling over one MDK in losses. 

\begin{figure}[!ht]
    \centering
    \includegraphics[width=.7\textwidth]{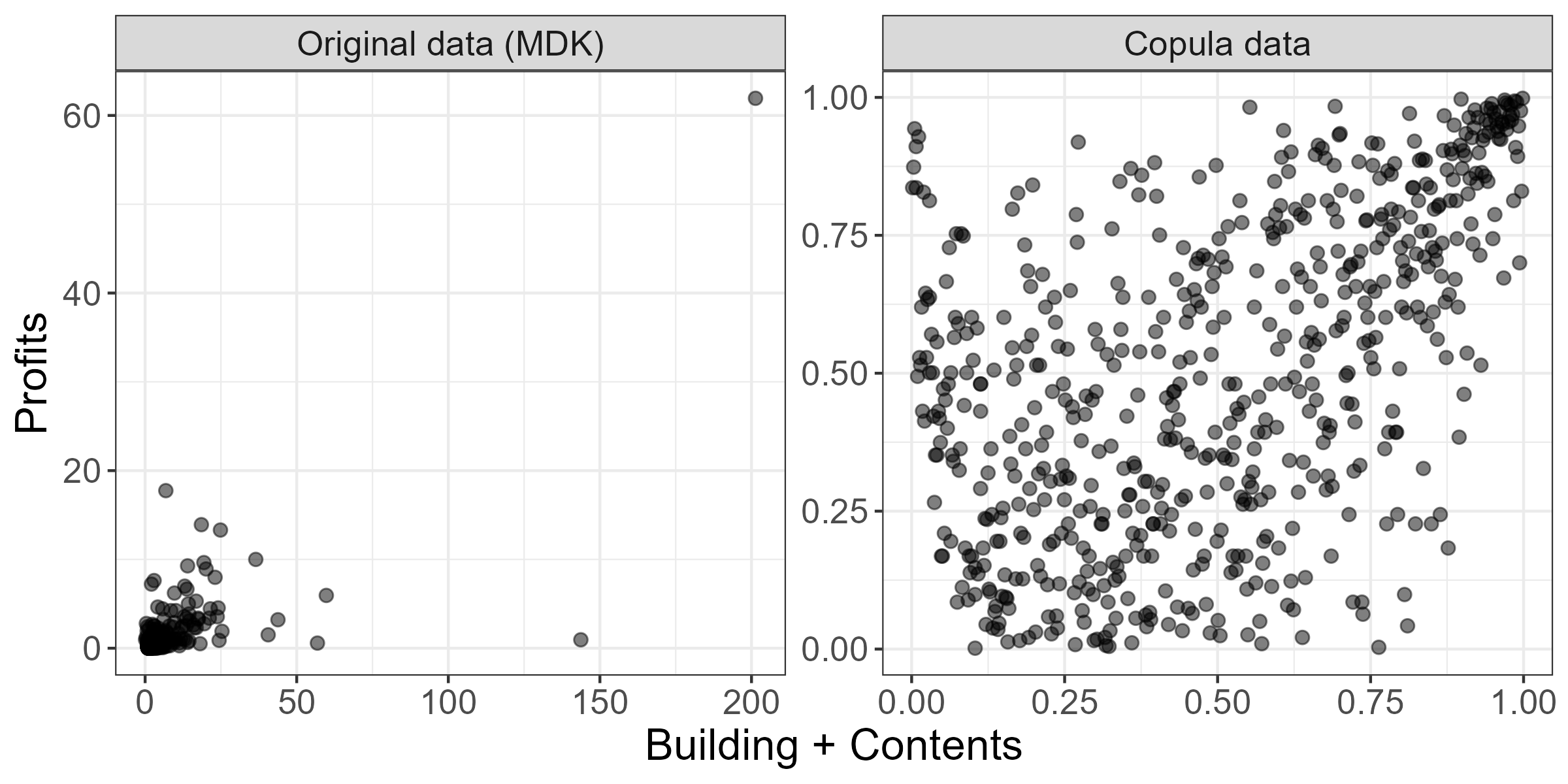}
    \caption{Losses to profits versus total material losses (the sum of losses to building and contents) in millions of Danish krone (MDK) (left panel) and as copula data (right panel).}
    \label{fig:danish_data}
\end{figure}

For the formal analysis, we fitted the PD copula to the Danish fire insurance data, and we compared the fit to those of a suite of Archimedean, extreme-value, and elliptical copulas implemented in the `copula' R package  \citep{copula}. The Archimedean copulas used for the comparison were the Clayton \citep{Clayton1978}, Gumbel \citep{Gumbel1960}, Frank \citep{Frank1979}, and Joe \citep{Joe1993} copulas \citep[also see][Table 4.2.1 and surrounding exposition]{Nelsen2006}. For the Clayton copula, its survival copula was fitted to the data instead. The extreme-value copulas were the Galambos \citep{Galambos1975}, Husler-Reiss \citep{Husler1989}, and Tawn \citep{Tawn1988} copulas. A bivariate Gaussian copula was also fitted. Parameters were estimated using the method-of-moments estimator based on inversion of Kendall's tau \citep[e.g.,][]{Genest1993}, which was calculated to be $0.361$ from the data. Table \ref{tab:summary_table_insurance} shows the estimates of the fitted copulas' parameters. 

\begin{table}[!ht]
    \centering
    \caption{Copulas fitted to the Danish fire insurance data, their type (AR for Archimedean, EV for extreme-value, and EL for elliptical), and their estimated parameters. The symbol (S) next to a copula's name indicates its survival copula was fitted. Goodness-of-fit tests \citep{Genest2009, Genest2011} are reported in the last two columns.\vspace{3pt}}
    \begin{tabular}{|ccc|cc|}
    \hline
         \multicolumn{3}{|c|}{Fitted copulas} & \multicolumn{2}{c|}{Goodness-of-fit test}\\
         \hline
         Copula & Type & Fitted parameter & Statistic & p-value \\
         \hline
         Power-divergence & AR & -0.647 & 0.036 & $0.078$\\
         Clayton (S) & AR & 1.129 & 0.074 & $<0.001$\\
         Frank & AR & 3.645 & 0.303 & $<0.001$ \\
         Joe & AR & 2.027 & 0.069 & $<0.001$\\
         Gumbel & AR/EV & 1.565 & 0.165 & $<0.001$\\
         Galambos & EV & 0.841 & 0.093 & $0.014$\\
         Husler-Reiss & EV & 1.282 & 0.101 & $0.007$\\
         Tawn & EV & 0.888 & 0.085 & $0.018$ \\
         Gaussian & EL & 0.537 &  0.313 & $<0.001$\\
         \hline
    \end{tabular}
    \label{tab:summary_table_insurance}
\end{table}

The parametric-bootstrap goodness-of-fit (PB-GOF) test described in \citet{Genest2009} was used to check whether the fitted Archimedean copulas and the Gaussian copula provided an adequate fit to the Danish fire insurance data. The goodness-of-fit test from \citet{Genest2011} was used for the extreme-value copulas. Almost all tests were performed using the `copula' R package. Only the PB-GOF test for the PD copula required a bespoke implementation of the algorithm in Appendix A of \citet{Genest2009}. Table \ref{tab:summary_table_insurance} reports the test statistics and $p$-values, which indicate that, at a 5\% significance level, only the fitted PD copula with $\hat\lambda = -0.647$ provides an adequate fit to the Danish fire insurance data. This is likely due to its moderate upper-tail dependence, shape of its copula density (which exists for $\hat\lambda = -0.647$ due to Theorem \ref{thm:absolutely_continuous}), and zero set that captures the region where small material losses and losses to profits do not occur together. 

\begin{figure}[!ht]
    \centering
    \includegraphics[width=\textwidth]{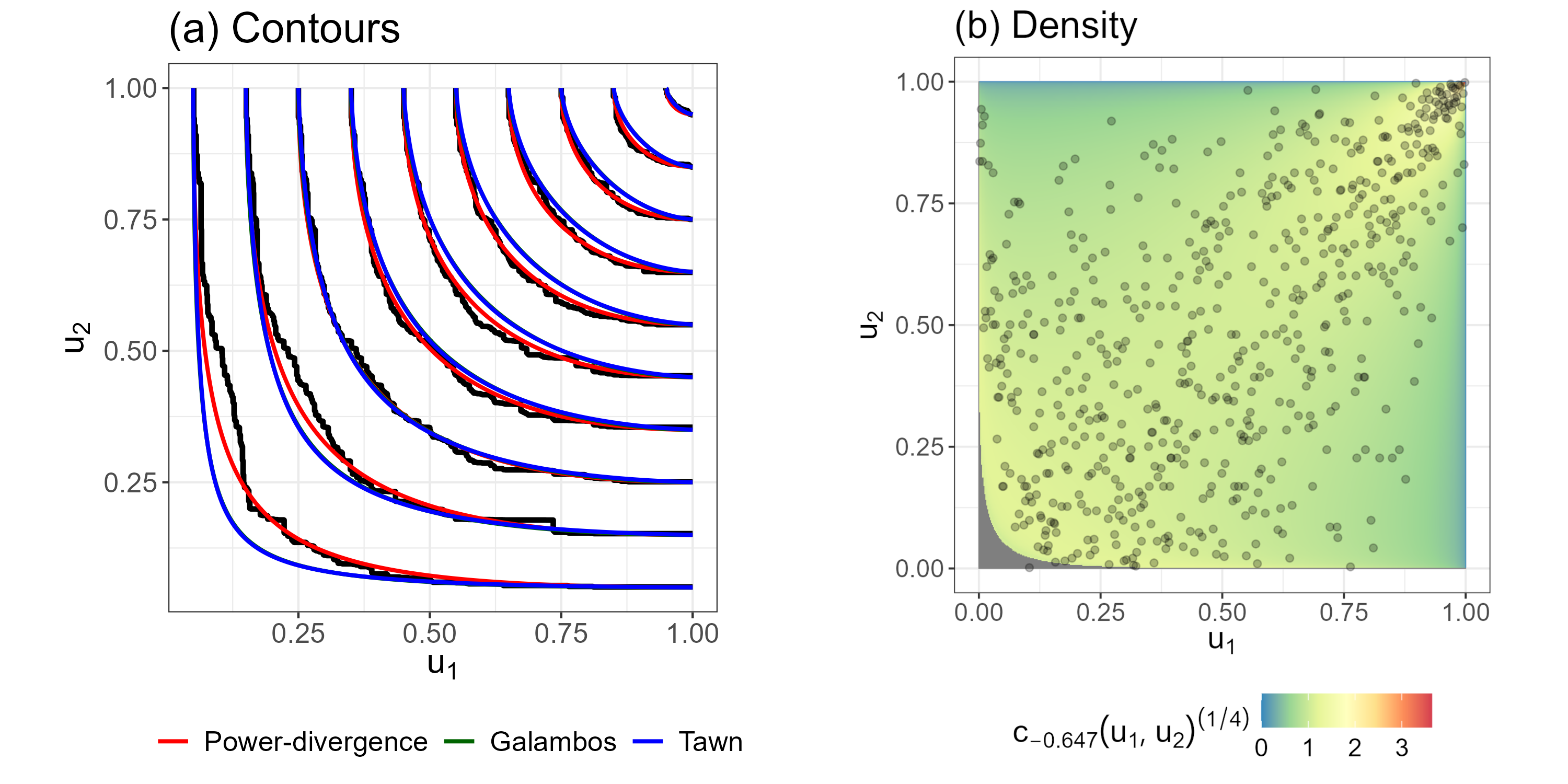}
    \caption{(a) Contours of the empirical copula (jagged black lines) and fitted power-divergence (PD) copula, Galambos copula, and Tawn copula (see legend). The contours are displayed for the level curves of the copulas at $\{0.05, 0.15, ..., 0.85, 0.95\}$. (b) The copula data superimposed on (the fourth root of) the copula density of the fitted PD copula. (The fourth-root transformation of the density is purely intended to aid visualisation.)}
    \label{fig:fitness}
\end{figure}

Visual confirmation of the PB-GOF results is provided by Fig. \ref{fig:fitness}. In particular, Fig. \ref{fig:fitness}(a) shows that the contours of the PD copula with $\lambda = -0.647$ line up closely with those from the empirical copula of the Danish fire insurance data. The contours of the Galambos and Tawn copulas are also shown because these two achieved the highest $p$-values in the PB-GOF test among the benchmark copulas. These copulas have virtually identical contours. Neither fit the empirical copula well, particularly in the lower tail. Fig. \ref{fig:fitness}(b) plots the copula density of the PD copula with $\hat\lambda = -0.647$ underneath the copula data. This emphasises the role of the upper-tail dependence and the zero curve in capturing the shape of the bivariate dependence structure. In Supplement S5, we also show that synthetic datasets simulated from the fitted PD copula with $\hat\lambda = -0.647$ resemble the original data, but synthetic datasets simulated from the Galambos and Tawn copulas do not. 

\section{Discussion and conclusion}\label{sec:discussion}

This paper has shown that the same functions that generate $\phi$ divergences also generate families of Archimedean copulas. Subsequently, we developed the family associated with the well known power divergences \citep{Cressie1984, Read1988}, called the power-divergence (PD) copulas. The properties of this family were extensively studied. Analysis of a Danish fire insurance dataset demonstrated the potential usefulness of the PD copulas, which achieved an adequate fit to the data when several Archimedean, extreme-value, and elliptical copulas did not. As a final note, Theorem \ref{prop:phi_as_copula_generator} implies there are Archimedean copulas associated with other $\phi$ divergences. The exploration of these copulas is left to future research.

\section*{Acknowledgements} 

HB was supported by Australian Research Council Future Fellowship FT190100374. ARP (postdoctoral research fellow) was partially supported by the same grant. ARP thanks Dr David Gunawan for a helpful discussion about the  applications of Archimedean copulas. 

\bibliographystyle{apalike}
\bibliography{PearseBondell_ArchCopula_Arxiv.bib}

\appendix

\setcounter{equation}{0}
\setcounter{figure}{0}
\setcounter{algorithm}{0}
\renewcommand{\theequation}{S\arabic{equation}}
\renewcommand{\thealgorithm}{S\arabic{algorithm}}
\renewcommand{\thefigure}{S\arabic{figure}}
\renewcommand{\thesection}{S\arabic{section}}
\renewcommand{\thesubsection}{S\arabic{section}.\arabic{subsection}}

\section*{{\Large Supplementary information}}

This is the Supplementary Information (`the Supplement') for the paper. It is structured as follows. Section S1 discusses an alternative definition of phi divergences (here labelled `$\varphi$ divergences') used by some researchers and its relationship to the one given in Definition \ref{def:phi_divergence} (`$\phi$ divergences'). Section S2 gives an example of a $\phi$ function that is not continuously differentiable on $(0, \infty)$. Section S3 presents exact formulas for the power-divergence (PD) copulas with $\lambda = -2$, $\lambda = -0.5$ and $\lambda = 1$. In Section S4, we show how to simulate from a bivariate PD copula. Section S5 presents additional results for the analysis of the Danish fire insurance data. Section S6 defines the incomplete exponential Bell polynomials and proves a lemma about them that is needed for the following section. Section S7 presents the proofs of all propositions in the main text. 

\section{Alternative definition of phi divergence} 

\subsection{A different form of phi ($\varphi$) divergence}

When defining a phi divergence of one probability distribution from another, alternative definitions have been used other than Definition \ref{def:phi_divergence}. For clarity, from this point on, we use the symbol $\phi$ to refer to $\phi$ divergences and $\phi$-divergence generators under Definition \ref{def:phi_divergence}. We use $\varphi$ to refer to $\varphi$ divergences under the alternative definition given below.

The following definition of a $\varphi$ divergence is sometimes used to define divergences of one probability distribution from another. 

\begin{definition}\label{def:suppinfo_alternative}
Let $P$ and $Q$ be absolutely continuous probability distributions admitting densities $p$ and $q$ with respect to Lebesgue measure. Let $\varphi(x)$ be a convex function for $x \geq 0$ that satisfies,
\begin{enumerate}[label=(\alph*)]
    \item $\varphi(1) = 0$, and
    \item $\varphi''(1) > 0$.
\end{enumerate}
Then, the $\varphi$ divergence of $P$ from $Q$ is defined as,
$$
D_\varphi(p||q) \equiv \int q(s)\times \varphi\!\left(\!\frac{p(s)}{q(s)}\!\right)~\d s,
$$
where $0\times\varphi(0/0) = 0$, and $0\times\varphi(v/0) = v \times \{\lim_{x \to \infty} \varphi(x)/x\}$. 
\end{definition}

When the probability densities $p$ and $q$ are used as arguments of the $\phi$ divergence of Definition \ref{def:phi_divergence}, it is straightforward to show that $D_\phi(p \lVert q) = \int q(s)\times \phi(p(s)/q(s))~\d s$ and $D_\varphi(p \lVert q) = \int q(s)\times \varphi(p(s)/q(s))~\d s$ are exactly equivalent. However, while Definition \ref{def:suppinfo_alternative} is only valid when $\int p(s)~\d s = 1$ and $\int q(s)~\d s =1$, Definition \ref{def:phi_divergence} is valid for non-negative functions that may not integrate to unity.

\subsection{Relationship to Definition \ref{def:phi_divergence}}

We can recover $\phi(x)$ from $\varphi(x)$ ($x \geq 0$) by noting the affine-invariance property of convex functions \citep[e.g.,][p. 56]{Amari2016}: That is, if $\varphi(x)$ is convex with $\varphi(1) = 0$ and $\varphi''(1) > 0$, the function $\phi^@(x) = \varphi(x) + c\times (x-1)$, with real-valued constant $c$, is also convex with $\phi^@(1) = 0$ and $(\phi^@)''(1) > 0$. Now, following, for example, \citet{Cressie2002}, take $c = -\varphi'(1)$, and define $\phi(x) \equiv  \varphi(x) + \varphi'(1)\times (1 - x)$. Compared to $\varphi(x)$, $\phi(x)$ has the additional property that $\phi'(1) = 0$. 

\subsection{Failure of $\varphi$ to be an Archimedean copula generator}

The function $\varphi(x),~x\geq 0,$ in Definition \ref{def:suppinfo_alternative} does not necessarily correspond to an Archimedean-copula generator, though some do. For example, $\varphi(x) = -\log(x)$ satisfies (a) and (b) of Definition \ref{def:suppinfo_alternative}, and it is a valid Archimedean-copula generator that corresponds to the product copula. However, this does not hold in general. Consider the counterexample below. 

Let $\lambda \in (-\infty, \infty)$. \citet{Cressie1984} consider the function, 
$$
\varphi_\lambda(x) \equiv \frac{x^{\lambda+1}-x}{\lambda(\lambda+1)},~\lambda \neq -1, 0,
$$
where $x\geq 0$, and $\varphi_0(x)\equiv \lim_{\lambda \to 0} \phi_\lambda(x)$ and $\varphi_{-1}(x) \equiv \lim_{\lambda \to -1}\varphi_\lambda(x)$. It is straightforward to verify that $\phi_\lambda(x)$ in \eqref{eqn:CR_phi_function} is $\phi_\lambda(x) \equiv  \varphi_\lambda(x) + \varphi_\lambda'(1)(1 -x)$ for $x\geq0$.

Now take $\lambda = 1$, which yields $\varphi_1(x) = 0.5(x^2-x)$. Then $\varphi_1(x)$ is not strictly decreasing on $x \in [0, 1]$, since its first derivative, $\varphi_1'(x) = x-0.5$, is negative for $x \in [0, 0.5)$, zero at $x = 0.5$, and positive for $x \in (0.5, 1]$. This is one example of a $\varphi$ function that does not satisfy the requirements of an Archimedean-copula generator given in Definition \ref{def:generator}.

\section{$\phi$ function that is not continuously differentiable}

Define the piecewise $\phi$ function, 
$$
\phi_{\pw}(x) =\begin{cases} 
0.25 - \exp\{2\} + \exp\{1/x\} & 0 \leq x < 0.5,\\
(1-x)^2 & 0.5 \leq x <\infty.
\end{cases}
$$ 
The function $\phi_{\pw}$ satisfies (a)-(c) of Definition \ref{def:phi_divergence}. It is strictly convex for all $x \in [0, \infty)$; it is equal to zero at $x = 1$; it has two derivatives at $x = 1$, with $\phi_\pw'(1) = 0$ and $\phi_\pw''(1)>0$; but it lacks a derivative at $x = 0.5$. After restricting $x \in [0, 1]$, the pseudoinverse of $\phi_\pw(x)$ is given by,
$$
\phi_{\pw}^{[-1]}(t) = \begin{cases}
    1-\sqrt{t} & 0 \leq t < 0.25,\\
    [\log(t -0.25 + \exp\{2\})]^{-1} & 0.25 \leq t < \infty.
\end{cases}
$$
This is a strict inverse, so $\phi^{[-1]}_{\pw}(t) = \phi_{\pw}^{-1}(t)$ for all $t\in[0,\infty)$, but it lacks a derivative at $t = 0.25$, where the two cases meet. Fig. \ref{fig:piecewise} clearly illustrates this point. Since $\phi_\pw^{-1}$ is not differentiable for all $t \in (0, \infty)$, $\phi^{-1}_\pw$ cannot be $d$-monotone on $t\in[0,\infty)$ for any $d \geq 3$ \citep{McNeil2009}. Therefore, $\phi_\pw$ does not generate a valid Archimedean copula in any dimensions $d \geq 3$.

\begin{figure}[ht]
    \centering
    \includegraphics[width=\textwidth]{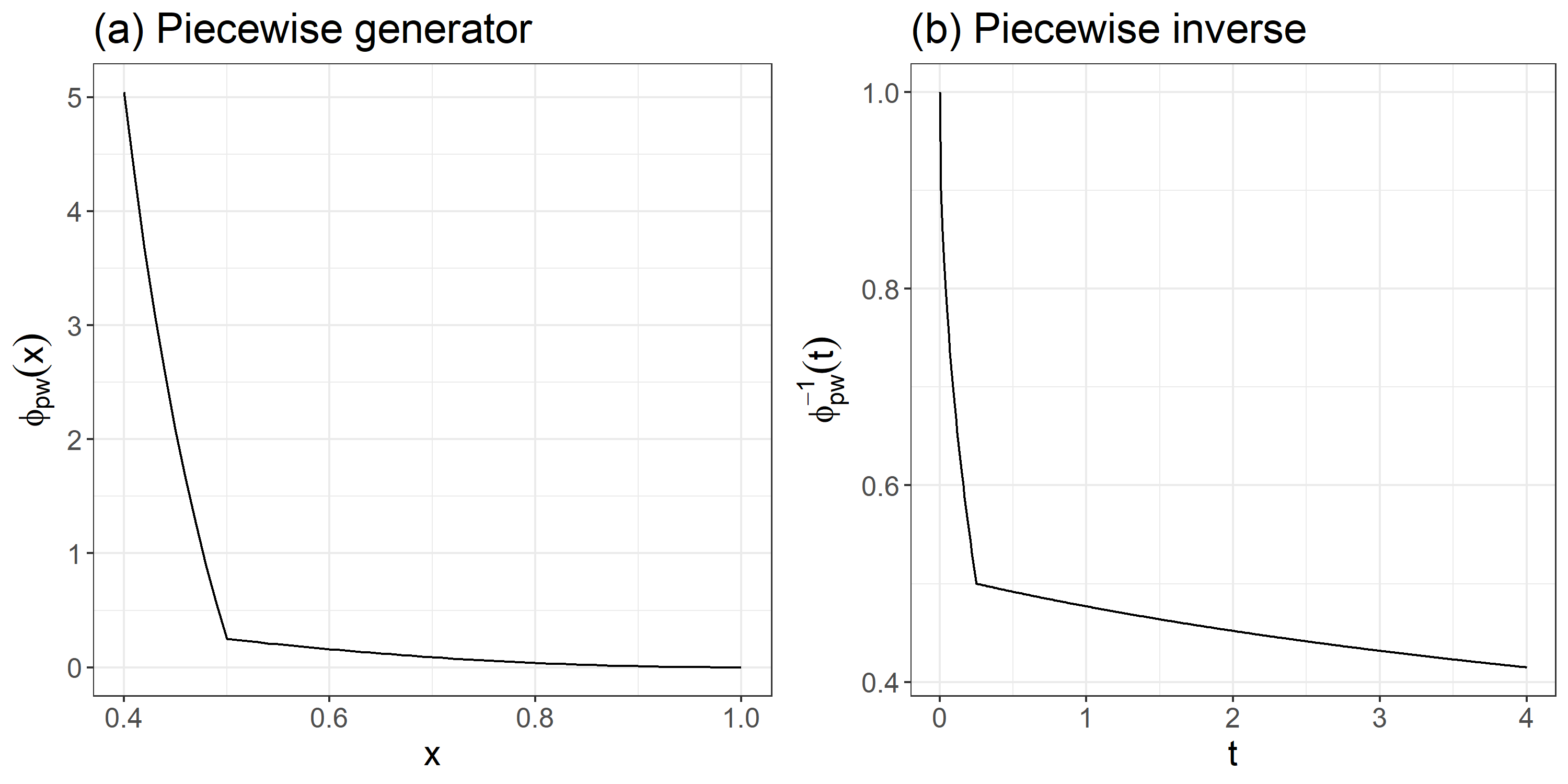}
    \caption{The left panel (a) shows the piecewise $\phi$ function, $\phi_{\pw}(x),~x\in [0,1]$. Note that $\phi_{\pw}(x)$  has only been plotted for $x\in [0.4,1]$ to `zoom in' on the point $x = 0.5$, where the function is continuous but lacks a derivative. The right panel (b) shows the (strict) inverse $\phi_{\pw}^{-1}(t)$ over $t\in [0, 4]$. }
    \label{fig:piecewise}
\end{figure}

\newpage

\section{Exact formulas for $\lambda = -2$, $\lambda = -0.5$ and $\lambda = 1$}

For $u_1,u_2\in[0,1]$, the PD copula with $\lambda = -2$ is,
\begin{align}
    C_{-2}(u_1,u_2) = 1 &+ 0.5(u_1^{-1} + u_1+u_2^{-1} + u_2 - 4)\nonumber\\
    &- \sqrt{0.25(u_1^{-1} + u_1+u_2^{-1} + u_2 - 4)^2 + u_1^{-1} + u_1+u_2^{-1} + u_2 - 4}.\label{eqn:Cn2}
\end{align}
The associated copula density, which exists for all $u_1,u_2\in [0, 1]$ due to Theorem \ref{thm:absolutely_continuous}, is denoted by $c_{-2}(u_1, u_2)$. The expression is not straightforward to simplify, but it can be obtained in closed form by substituting \eqref{eqn:Cn2} into \eqref{eqn:PD_copula_dens}. 

\begin{figure}[!ht]
    \centering
    \includegraphics[width=.78\textwidth]{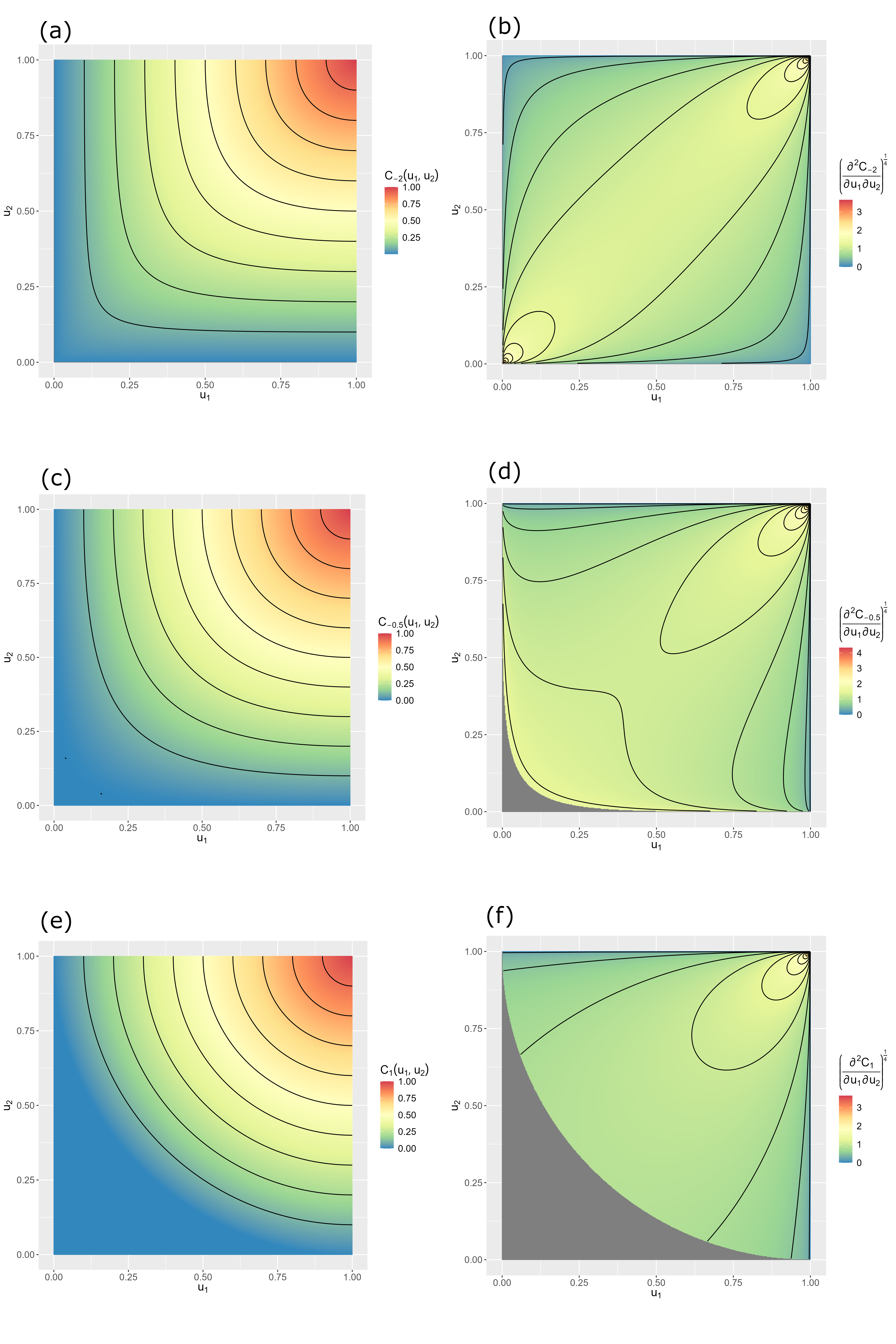}
    \caption{Left column: (a) Power-divergence (PD) copula with $\lambda = -2$; see \eqref{eqn:Cn2}. (c) PD copula with $\lambda = -0.5$; see \eqref{eqn:C05}. (e) PD copula with $\lambda = 1$; see \eqref{eqn:C1}. Right column: (b) The fourth root of the PD copula density with $\lambda = -2$. (d) The fourth root of the PD copula density with $\lambda = -0.5$; see \eqref{eqn:c05}. (f) The fourth root of the derivative $\partial C_1(u_1, u_2)/\partial u_1\partial u_2$ for $u_1, u_2 \in [0,1]$; see \eqref{eqn:c1}. The fourth-root transformations in (b), (d), and (f) are used for visualisation purposes only.}
    \label{fig:closed_form}
\end{figure}

The PD copula with $\lambda = -0.5$ is,
\begin{align}
    &C_{-0.5}(u_1, u_2)\nonumber\\
    &=\begin{cases}
    3 - 2(\sqrt{u_1} + \sqrt{u_2}) + (u_1 + u_2) & 0 \leq 2 - 2(\sqrt{u_1}+\sqrt{u_2}) + (u_1 + u_2) < 1,\\
     ~~-~2\sqrt{2-2(\sqrt{u_1} + \sqrt{u_2}) + (u_1 + u_2)}& \\
    0 & 1 \leq 2 - 2(\sqrt{u_1}+\sqrt{u_2}) + (u_1 + u_2)< \infty.
\end{cases}\label{eqn:C05}
\end{align}
The copula $C_{-0.5}$ is absolutely continuous by Theorem \ref{thm:absolutely_continuous}. For $0 \leq 2 - 2(\sqrt{u_1}+\sqrt{u_2}) + (u_1 + u_2) \leq 1$, the associated copula density is given by,
\begin{align}
    c_{-0.5}(u_1, u_2) =\frac{(1-(\sqrt{u_1})^{-1})(1-(\sqrt{u_2})^{-1})}{(2-2(\sqrt{u_1} + \sqrt{u_2})+u_1+u_2)^{3/2}}.\label{eqn:c05}
\end{align}
On the other hand, with $\lambda = 1$, the PD copula is,
\begin{equation}
C_1(u_1, u_2) = \max\!\left\{1 - \sqrt{u_1^2 - 2u_1 + u^2_2 - 2u_2 + 2}, 0\right\},~u_1,u_2\in[0,1].\label{eqn:C1}    
\end{equation}
The copula $C_1(u_1, u_2)$ has a singular part and an absolutely continuous part. The singular part is supported on the zero curve (the curve in $[0,1]^2$ traced by $u^2_1 - 2 u_1 + u_2^2 - 2 u_2 + 2 = 1$), and it contains 50\% of the probability mass of the bivariate distribution by Theorem \ref{thm:absolutely_continuous}. On the absolutely continuous part of the copula, over the part of the unit square where $u^2_1 - 2 u_1 + u_2^2 - 2 u_2 + 2 < 1$, the derivative $\partial^2C_{1}(u_1, u_2)/\partial u_1 \partial u_2$ exists and is given by,
\begin{equation}
\frac{\partial^2 C_1}{\partial u_1 \partial u_2} = \frac{(u_1 - 1)(u_2 - 1)}{(u_1^2 - 2 u_1 + u_2^2 - 2u_2 + 2)^{3/2}}.\label{eqn:c1}   
\end{equation}

Fig. \ref{fig:closed_form} displays the copulas and derivatives for $\lambda = -0.5$ and $\lambda = 1$. Recall that the PD copulas with $\lambda = -2$ and $\lambda = -0.5$ are absolutely continuous, but the PD copula with $\lambda = 1$ is not. Therefore, Fig. \ref{fig:closed_form}(b) and \ref{fig:closed_form}(d) show copula densities whereas, strictly speaking, Fig. \ref{fig:closed_form}(f) is not a density.

\section{Bivariate simulation}\label{sec:bivariate_simulation}

Algorithm \ref{alg:sim_bivariate} below provides details of the conditional distribution method \citep[e.g., ][Sec 2.9]{Nelsen2006} for simulating from bivariate PD copulas. 

\begin{algorithm}
\caption{The conditional distribution method for simulating from the bivariate power-divergence copula with $\lambda \in (-\infty, \infty)$.}\label{alg:sim_bivariate}
\begin{algorithmic}[1]
\State Let $\lambda \in (-\infty, \infty)$.
\State Simulate $u_1$ and $t$ independently from the uniform distribution on $[0, 1]$.
\If{$\lambda = 0$}
\State Compute $u_2 = \phi_{0}^{[-1]}\!\left(\phi_0\!\left(u_1^{1/t}\right) - \phi_0(u_1)\right)$ using \eqref{eqn:CR_phi_function} and \eqref{eqn:pseudoinverse_0}.
\Else
\If{$1 + t^{-1}(u_1^\lambda - 1) < 0$}
    \State Compute $u_2 = \phi_\lambda^{[-1]}(\phi_\lambda(0) - \phi_\lambda(u_1))$ using \eqref{eqn:CR_phi_function}, \eqref{eqn:pseudoinverse_neg1}, and \eqref{eqn:zeros_of_pseudopolynomial} as appropriate. 
\Else
    \State Compute $u_2 = \phi_\lambda^{[-1]}\!\left(\phi_\lambda\!\left((1 + t^{-1}(u_1^\lambda - 1))^{1/\lambda}\right) - \phi_\lambda(u_1)\right)$; see \eqref{eqn:CR_phi_function}, \eqref{eqn:pseudoinverse_neg1}, and \eqref{eqn:zeros_of_pseudopolynomial}.
\EndIf
\EndIf
\State Return the pair $(u_1, u_2)$. Discard $t$. 
\end{algorithmic}
\end{algorithm}

\section{Simulating new insurance data}\label{sec:simulating_insurance_data}

Section \ref{sec:applied} shows that the PD copula with $\hat\lambda = -0.647$ appears to fit the Danish fire insurance data well. Downstream analyses using the fitted PD copula may include the calculation of joint exceedance probabilities for extreme material losses and losses to profits. The PD copula could also be used to simulate future insurance claims, with the fitted model enabling simulation-based decision-making. Fig. \ref{fig:insurance_sims} suggests that the fitted PD copula with $\hat\lambda = -0.647$ can be used to simulate realistic datasets, while random variates from the Galambos and Tawn copulas appear to exhibit the correct upper-tail dependence, but they do not resemble the data in the lower tail.

\begin{figure}[!ht]
    \centering
    \includegraphics[width=0.65\textwidth]{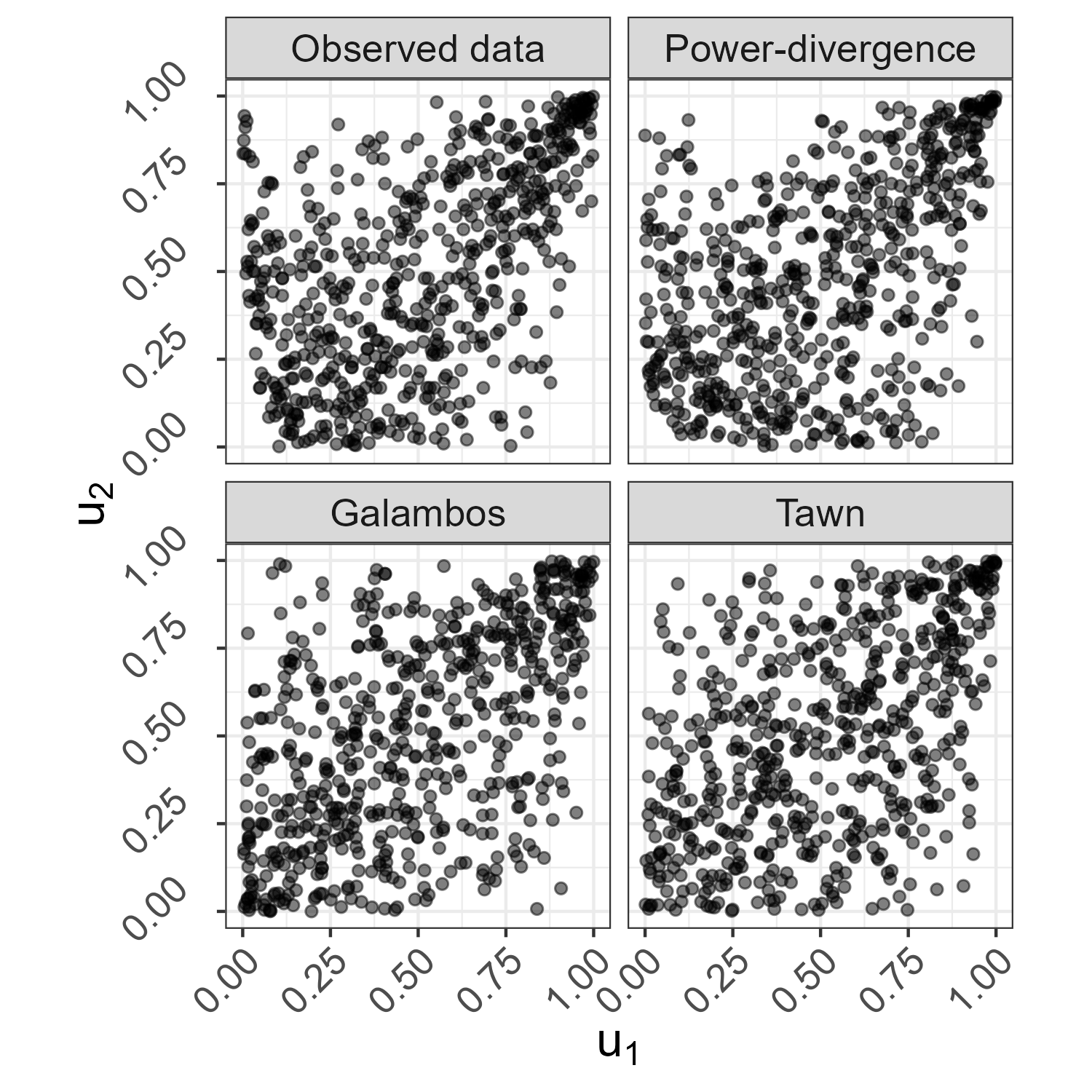}
    \caption{The observed Danish fire insurance data and datasets of the same size ($n=616$) simulated from the fitted power-divergence, Galambos, and Tawn copulas.}
    \label{fig:insurance_sims}
\end{figure}

\section{Incomplete exponential Bell polynomials}\label{sec:Bell_polynomials}

Here, we define incomplete exponential Bell polynomials \citep{Bell1934}. See the definition below, which follows \citet[pp. 133-137]{Comtet1974}. From this point, we simply refer to incomplete exponential Bell polynomials as `Bell polynomials'. 

\begin{definition}\label{def:Bell_polynomial}
    Let $k$ and $j$ be positive integers with $k \geq j$. For variables $x_1, ..., x_{k-j+1}$ and non-negative integers $c_{1}, c_{2}, ..., c_{k-j+1}$, the incomplete exponential Bell polynomial $\gB_{k,j} \equiv B_{k,j}(x_1, ..., x_{k-j+1})$ is defined as,
    \begin{align}
    &\gB_{k,j} \equiv B_{k,j}(x_1, ..., x_{k-j+1})\nonumber\\
    &~~~~~= \sum \frac{k!}{c_1!c_2 !\cdots c_{k-j+1}!}\left(\frac{x_1}{1!}\right)^{c_1} \left(\frac{x_2}{2!}\right)^{c_2}\cdots \left(\frac{x_{k-j+1}}{(k-j+1)!}\right)^{c_{k-j+1}},\label{eqn:Bell_polynomial}
    \end{align}
    where the sum is taken over all (non-negative integer) solutions of the system of equations,
    \begin{align}
        c_1 + c_2 + \cdots + c_{k-j+1} &= j,\label{eqn:sum_to_j}\\
        1c_1 + 2c_2 + \cdots + (k-j+1) c_{k-j+1} &= k.\label{eqn:total_weight_sum_to_k}
    \end{align}
\end{definition}

The lemma below establishes a property of Bell polynomials that is needed to prove Lemma \ref{lem:pd_copulas_completely_monotone}. The following proposition is proved immediately, but the proofs of Lemma \ref{lem:pd_copulas_completely_monotone} and the other propositions from the main text are provided in the next section.

\begin{lemma}\label{lem:Bell_polynomials}
    For positive integers $k,j$ where $k \geq j$, let $x_1, ..., x_{k-j+1}$ be variables that satisfy $(-1)^lx_l \geq 0$ for $l = 1, ..., k-j+1$. Define $\gB_{k,j} \equiv B_{k,j}(x_1, ..., x_{k-j+1})$ as in Definition \ref{def:Bell_polynomial}. Then, for all $j = 1, ..., k$, $\gB_{k,j} \geq 0$ if $k$ is even, and $\gB_{k,j} \leq 0$ if $k$ is odd.
\end{lemma}

\begin{proof}
    From Definition \ref{def:Bell_polynomial}, each summand in $\gB_{k,j}$ is,
    \begin{equation}
            \frac{k!}{c_1!c_2 !\cdots c_{k-j+1}!}\prod_{l=1}^{k-j+1}\left(\frac{x_l}{l!}\right)^{c_l}.\label{eqn:crucial_product}
    \end{equation}
    The term \eqref{eqn:crucial_product} is non-negative if $k$ is even and non-positive if $k$ is odd. To see this, observe that the number $l\times c_l$ is even if and only if $c_l$ is even or $l$ is even (hence $x_l \geq 0$). 
    It follows that, if $l \times c_l$ is even, then $(x_{l}/l!)^{c_l} \geq 0$ since either $x_{l} \geq 0$ or its exponent is even. As for the only remaining case, observe that the number $l\times c_l$ is odd if and only if $c_l$ is odd and $l$ is odd (hence $x_l \leq 0$). Hence, if $l\times c_l$ is odd, then $(x_{l}/l!)^{c_l} \leq 0$ since odd powers of a non-positive value are also non-positive. 

    Now, from \eqref{eqn:total_weight_sum_to_k}, we see that $\sum_{l=1}^{k-j+1} (l\times c_l) = k$. If the integer $k$ is even, then the sum in \eqref{eqn:total_weight_sum_to_k} must consist of any number of terms where $l\times c_l$ is even, and an even number of terms where $l\times c_l$ is odd. Therefore, the product \eqref{eqn:crucial_product} must be non-negative if $k$ is even because it will be a product of non-negative terms (since $(x_l/l!)^{c_l} \geq 0$ when $l\times c_l$ is even) and an even number of non-positive terms (since $(x_l/l!)^{c_l} \leq 0$ when $l\times c_l$ is odd). On the other hand, if $k$ is odd, then the sum in \eqref{eqn:total_weight_sum_to_k} consists of any number of terms where $l\times c_l$ is even, and an odd number of non-positive terms where $l\times c_l$ is odd. This implies the product \eqref{eqn:crucial_product} is a product of some non-negative terms as well as, crucially, an odd number of non-positive terms (since $(x_l/l!)^{c_l} \leq 0$ if $l \times c_l$ is odd). This implies that, if $k$ is odd, \eqref{eqn:crucial_product} is non-positive. 
    
    Finally, returning to the definition of $\gB_{k,j}$ in \eqref{eqn:Bell_polynomial}, we see that \eqref{eqn:Bell_polynomial} is a sum of \eqref{eqn:crucial_product}. Since only the value of $k$, not $j$, determines whether \eqref{eqn:crucial_product} is non-negative or non-positive, it follows that, for any given $k$, \eqref{eqn:Bell_polynomial} is a sum of only non-negative terms if $k$ is even and a sum of only non-positive terms if $k$ is odd. Therefore, if $(-1)^l x_l \geq 0$ for all $l = 1, ..., k-j+1$, then, for all $j = 1, ..., k$,  $\gB_{k,j} \geq 0$ if $k$ is even and $\gB_{k,j} \leq 0$ if $k$ is odd, as required.
\end{proof}

\section{Proofs of results in the main text}\label{sec:proofs}

\begin{proof}[Proof of Proposition \ref{prop:phi_as_copula_generator}]
    By Definition \ref{def:phi_divergence}, $\phi(1) = 0$, and $\phi(x)$ is a convex function for all $x \geq 0$ (strictly convex at $x = 1$). As for being strictly decreasing over $x\in [0, 1]$, the convexity of $\phi$, and the properties $\phi'(1) = 0$ and $\phi''(1) > 0$, together imply that $\phi(x)$ is minimised uniquely at $x = 1$ and is strictly decreasing over $x \in [0, 1]$. Therefore, the proposition follows.
\end{proof}

\begin{proof}[Proof of Proposition \ref{prop:unique_solution}]
    First, when $t = 0$, \eqref{eqn:zeros_of_pseudopolynomial} becomes $0 = x^{\lambda + 1} - (\lambda + 1)x + \lambda$ for $\lambda \neq -1, 0$. Clearly, $x = 1$ is a solution. To see that it is the \textit{unique} solution in $[0, 1]$, define $g_\lambda(x) = x^{\lambda + 1} - (\lambda + 1)x + \lambda$. The first derivative with respect to $x$ is $g_\lambda'(x) = (\lambda + 1)(x^\lambda - 1)$. For all $\lambda \neq -1, 0$, $g_\lambda'(1) = 0$. For $x\in[0,1)$, $g_\lambda'(x) < 0$ if either $\lambda > 0$ or $\lambda < -1$; if $-1 < \lambda < 0$, then $g_\lambda'(x) > 0$. This means $g_\lambda(x)$ is strictly monotone decreasing over $x\in[0,1)$ for $\lambda > 0$ and $\lambda < -1$, and it is strictly monotone increasing over $x \in [0,1)$ for $0 < \lambda < -1$. Since $x = 1$ satisfies $g_\lambda(1) = 0$, and since $g_\lambda(x)$ is either strictly decreasing or strictly increasing over $x\in[0,1)$, it follows that $x = 1$ is the only value of $x \in [0,1]$ that satisfies \eqref{eqn:zeros_of_pseudopolynomial} when $t = 0$.   

    Now let $t\in (0, \phi_\lambda(0))$ be fixed. Define the function $h_\lambda(x) = x^{\lambda+1}-(\lambda + 1)x + \lambda - \lambda(\lambda + 1)t$ for $x \in [0, 1]$ and $\lambda \neq -1, 0$. The derivative of $h_\lambda(x)$ with respect to $x$ is $h'_\lambda(x) = (\lambda + 1)\{x^\lambda - 1\}$, and $h'_\lambda(1)=0$ for all $\lambda \neq -1, 0$. With a slight abuse of notation, write $h_{a < \lambda <b}(x)$ and $h'_{a < \lambda < b}(x)$ to make statements about the values of $h_\lambda(x)$ and its derivative over an interval of $\lambda$ values, namely $a < \lambda < b$. Now consider the following three cases:
    \begin{enumerate}
        \item When $-\infty < \lambda < -1$ and $t \in (0, \infty)$, note that $h_{-\infty < \lambda < -1}(x)$ is a strictly monotone decreasing function over $x\in [0,1)$ since $h'_{-\infty < \lambda < -1}(x) < 0$ for all $x\in [0, 1)$. Further, $\lim_{x \downarrow 0} h_{-\infty < \lambda < -1}(x) = \infty$ and $h_{-\infty < \lambda < -1}(1) = -\lambda(\lambda + 1)t$, which is negative since, here, $\lambda(\lambda+1) > 0$ and $t \in (0, \infty)$. Being strictly decreasing over $x\in[0,1)$ with opposite signs at $x = 0$ and $x = 1$, the Intermediate Value Theorem (IVT) implies that $h_{-\infty < \lambda < -1}(x_0) = 0$ for a unique $x_0$ in the interval $[0, 1)$. 
        
        \item When $-1 < \lambda < 0$  and $t \in (0,1/(\lambda +1))$, note that $h_{-1 < \lambda < 0}(0) < 0$ and $h_{-1 < \lambda < 0}(1) > 0$ (i.e, the function has opposite signs at $x = 0$ and $x = 1$), while the derivative $h'_{-1 < \lambda < 0}(x) > 0$ for all $x \in [0, 1)$, meaning that $h_{-1 < \lambda < 0}(x)$ is a strictly increasing function of $x \in [0, 1)$. Together with the IVT, these facts imply that there exists a unique value $x_0\in [0,1)$ such that $h_{-1 < \lambda < 0}(x_0) = 0$. 
        
        \item When $0 < \lambda <\infty$ and $t \in (0,1/(\lambda + 1))$, we have $h_{0 < \lambda < \infty}(0) > 0$ and $h_{0 < \lambda < \infty}(1) < 0$ (i.e., the function has opposite signs at $x = 0$ and $x = 1$). The derivative of $h_{0 < \lambda < \infty}(x)$ is $h_{0 < \lambda < \infty}'(x) < 0$ for all $x \in [0, 1)$, so $h(x)$ is a strictly decreasing function of $x \in [0,1)$. Together with the IVT, these facts imply that there exists a unique value $x_0 \in [0, 1)$ for which $h_{0 < \lambda < \infty}(x_0) = 0$. 
    \end{enumerate}
    Considering all cases together, the proposition follows. 
\end{proof}

\begin{proof}[Proof of Theorem \ref{thm:zero_set}]
    Recall from \eqref{eqn:phi_lambda_zero} that $\phi_\lambda(0) = \infty$ for $\lambda \leq -1$ and $\phi_\lambda(0) = 1/(\lambda + 1)$ for $\lambda > -1$. Since the zero curve of the Archimedean copula $C_\lambda(u_1, u_2) \equiv C(u_1,u_2; \phi_\lambda)$ is traced by $(u_1,u_2)\in [0,1]^2$ that satisfy $\phi_\lambda(u_1) + \phi_\lambda(u_2) = \phi_\lambda(0)$, and since $\phi_\lambda(u) < \infty$ for all $u \in (0, 1]$, it is obvious that only elements of the set $\{(u_1, 0), (0, u_2): u_1, u_2 \in [0, 1]\}$ satisfy $\phi_\lambda(u_1) + \phi_\lambda(u_2)=\phi_\lambda(0)$ when $\lambda \leq -1$. This shows the PD copulas have a zero set consisting of only the lines $\{(u_1, 0), (0, u_2): u_1, u_2 \in [0,1]\}$ for $\lambda \leq -1$. On the other hand, when $\lambda > -1$, substitute $\phi_\lambda(x) = (\lambda(\lambda + 1))^{-1}(x^{\lambda + 1} - x + \lambda(1 - x))$ into $\phi_\lambda(u_1) + \phi_\lambda(u_2)=\phi_\lambda(0)$ to obtain the following equation for the zero curve:
    $$
    \frac{u_1^{\lambda + 1} - u_1 + u_2^{\lambda + 1} - u_2}{\lambda} + 2 - u_1 - u_2 = 1.
    $$
    The solutions are non-trivial, and this proves the second part of the theorem. The final part of the theorem follows from the fact that,
    $$
    \lim_{\lambda \to \infty} \frac{u_1^{\lambda + 1} - u_1 + u_2^{\lambda + 1} - u_2}{\lambda} + 2 - u_1 - u_2= 2 - u_1 - u_2 = 1,
    $$
    since $u_1, u_2 \in [0, 1]$. Rearranging this equation yields $u_1 + u_2 = 1$, which defines a straight line from $(1, 0)$ to $(0, 1)$ in the unit square. Hence the zero set is the triangle defined by the vertices $(0, 0)$, $(0, 1)$, and $(1, 0)$. Now, the theorem follows.
\end{proof}

\begin{proof}[Proof of Theorem \ref{thm:absolutely_continuous}]
    From \citet[Thm 4.3.3]{Nelsen2006}, it is sufficient to analyse the limit $-\{\lim_{s\downarrow 0} \phi_\lambda(s)/\phi_\lambda'(s)\}$. If this limit is zero, the copula is absolutely continuous. Otherwise, it has a singular component supported on the zero curve, with C-measure (i.e., probability mass) given by the value of the limit. We address this problem in five cases, namely when $\lambda < -1$, $\lambda = -1$, $-1 < \lambda < 0$, $\lambda = 0$, and $\lambda > 0$. When $\lambda < -1$, the limit,
        $$
        -\lim_{s\downarrow 0} \left\{\frac{\phi_\lambda(s)}{\phi_\lambda'(s)}\right\} = -\frac{1}{\lambda+1}\left\{\lim_{s\downarrow 0} \frac{s^{\lambda+1} - s + \lambda(1-s)}{s^\lambda-1}\right\},
        $$
    can be evaluated by two applications of L'H\^opital's rule. Hence, we have,
        $$
        -\left\{\lim_{s\downarrow 0} \frac{\phi_\lambda(s)}{\phi_\lambda'(s)}\right\} = -\frac{1}{(\lambda-1)}\left\{\lim_{s\downarrow 0} s\right\} = 0.
        $$
    When $\lambda = -1$, the limit, 
    $$
    -\left\{\lim_{s \downarrow 0} \frac{\phi_{-1}(s)}{\phi_{-1}'(s)}\right\} = -\left\{\lim_{s \downarrow 0} \frac{s - 1 - \log(s)}{1 - s^{-1}}\right\}, 
    $$
    can be evaluated after one application of L'H\^opital's rule, which then reveals,
    $$
    - \left\{\lim_{s \downarrow 0} \frac{\phi_{-1}(s)}{\phi_{-1}'(s)}\right\} = - \left\{\lim_{s \downarrow 0} \frac{1-s^{-1}}{s^{-2}}\right\} = -\left\{\lim_{s\downarrow 0} s^2 - s\right\} = 0.
    $$
    When $-1 < \lambda < 0$, the limit is,
        $$
        -\left\{\lim_{s\downarrow 0} \frac{\phi_\lambda(s)}{\phi_\lambda'(s)}\right\} = -\frac{1}{\lambda+1}\left\{\lim_{s\downarrow 0} \frac{s^{\lambda+1} - s + \lambda(1-s)}{s^\lambda-1}\right\} = 0,
        $$
        by inspection.
    When $\lambda = 0$, the limit is,
    $$
        -\left\{\lim_{s\downarrow 0} \frac{\phi_0(s)}{\phi_0'(s)}\right\} = -\left\{\lim_{s \downarrow 0} \frac{1 - s + s \log(s)}{\log(s)}\right\} = -\left\{\lim_{s \downarrow 0}\frac{1-s}{\log(s)} + \lim_{s\downarrow 0} s\right\} = 0.
    $$
    When $\lambda > 0$, we can evaluate the limit by inspection. That is,
        $$
        -\left\{\lim_{s\downarrow 0} \frac{\phi_\lambda(s)}{\phi_\lambda'(s)}\right\} = -\frac{1}{\lambda+1}\left\{\lim_{s\downarrow 0} \frac{s^{\lambda+1} - s + \lambda(1-s)}{s^\lambda-1}\right\} = \frac{\lambda}{\lambda + 1},
        $$
    which is always in the interval $(0, 1)$.
    
    Combining the analysis from all five cases reveals that,
    $$
    -\left\{\lim_{s\downarrow 0} \frac{\phi_\lambda(s)}{\phi_\lambda'(s)}\right\} = \begin{cases}
      0  &\lambda \leq 0,\\
      \frac{\lambda}{\lambda+1}  & \lambda > 0.
    \end{cases}
    $$
    Theorem 1 of \citet{Genest1986} then implies that the power-divergence copulas are absolutely continuous for $\lambda \leq 0$ and have a singular component supported on the zero curve with C-measure $\lambda/(\lambda+1)$ for $\lambda > 0$. Hence the theorem follows.
\end{proof}

\begin{proof}[Proof of Theorem \ref{thm:ordering}]
    The PD copula has the generator $\phi_\lambda(x)$ in  \eqref{eqn:CR_phi_function} with $\lambda \in (-\infty,\infty)$, which is continuously differentiable over $x \in (0, 1)$. Then, per \citet[Cor. 4.4.6]{Nelsen2006}, it suffices to show that, for $-\infty < \lambda_1 < \lambda_2 < \infty$, $g(x; \lambda_1, \lambda_2)=\phi_{\lambda_2}'(x)/\phi_{\lambda_1}'(x)$ is a non-decreasing function over $x \in (0, 1)$. This involves verifying that the derivative of $g(x; \lambda_1, \lambda_2)$ is non-negative in several cases. 

    Begin by noting that, for all $\lambda_1 < \lambda_2$ and $x \in (0, 1)$,
    $$
    g'(x; \lambda_1, \lambda_2) = \frac{\phi_{\lambda_2}''(x)\phi_{\lambda_1}'(x) - \phi_{\lambda_1}''(x)\phi_{\lambda_2}'(x)}{[\phi_{\lambda_1}'(x)]^2},
    $$
    by the quotient rule. Then also note that, for given $\lambda \in (-\infty, \infty)$ and $x \geq 0$ (hence also for $x \in (0,1)$), the function $\phi_\lambda(x)$ is non-negative, and its first two derivatives with respect to $x$ are given by,
    $$
    \phi_\lambda'(x) = \begin{cases}
        (x^\lambda - 1)/\lambda & \lambda \neq -1, 0,\\
        \log(x) & \lambda = 0,\\
        1 - x^{-1} & \lambda = -1,
    \end{cases}~~~\text{and}~~~    \phi_\lambda''(x) = \begin{cases}
        x^{\lambda - 1}& \lambda \neq -1, 0,\\
        x^{-1} & \lambda = 0,\\
        x^{-2} & \lambda = -1.\\
    \end{cases}
    $$
    These facts are needed in the sequel, where we examine the derivative of $g(x; \lambda_1, \lambda_2) = \phi_{\lambda_2}'(x)/\phi_{\lambda_1}'(x)$ for various combinations of $\lambda_1$ and $\lambda_2$.  

    If $\lambda_1 = -1$ and $\lambda_2 = 0$, we have,
    $$
    g'(x; -1, 0) = \frac{x - 1 - \log(x)}{(x-1)^2} = \frac{\phi_{-1}(x)}{(x-1)^2}.
    $$
    The derivative is always non-negative for $x \in (0,1)$ because the numerator is non-negative for $x \geq 0$ by definition, and the denominator is positive for $x \in (0, 1)$. Hence, $g(x; -1, 0)$ is non-decreasing over $x \in (0, 1)$.  

    If $\lambda_1 = -1$ and $\lambda_2 > -1$, we have to deal with two cases, namely $\lambda_2$ = 0 and $\lambda_2 \neq 0$. We have already dealt with the case where $\lambda_2 = 0$ above. Then, when $\lambda_2 \neq 0$, 
    $$
    g'(x; -1, \lambda_2) = \frac{x^{\lambda_2+1} - (\lambda_2^{-1}+1)x^{\lambda_2} + \lambda_2^{-1}}{(x-1)^2} = \frac{\lambda_2^{-1}(\lambda_2^{-1} + 1)\phi_{\lambda_2^{-1}}(x^{\lambda_2})}{(x-1)^2}.
    $$
    The denominator is obviously positive for $x \in (0, 1)$. The numerator is non-negative because $\phi_{\lambda}(x) \geq 0$ for all $\lambda \in (-\infty, \infty)$ and $x \geq 0$ by definition, and $\lambda_2^{-1}(1+\lambda_2^{-1}) = \lambda_2^{-2}(\lambda_2 + 1) > 0$ since $\lambda_2 > -1$. Hence, the derivative $g'(x; -1, \lambda_2)$ is non-negative for $x \in (0, 1)$. Therefore, $g(x; -1,\lambda_2)$ is non-decreasing over $x \in (0, 1)$.

    If $\lambda_1 = 0$ and $\lambda_2 > 0$, we have,
    $$
    g'(x; 0, \lambda_2) = \frac{(\lambda_2x)^{-1}(x^{\lambda_2}\log(x^{\lambda_2}) - x^{\lambda_2} + 1)}{\log(x)^2} = \frac{(\lambda_2x)^{-1}\phi_0(x^{\lambda_2})}{\log(x)^2}.
    $$
    The numerator is non-negative since $\lambda_2x>0$ and $\phi_0(x^{\lambda_2}) \geq 0$ for $x\in (0,1)$ and $\lambda_2 > 0$. The denominator is positive since $\log(x)^2 > 0$ for $x\in(0,1)$. Therefore, $g'(x; 0, \lambda_2) \geq 0$ for $x\in(0,1)$, which implies that $g(x; 0, \lambda_2)$ is a non-decreasing function over $x \in (0, 1)$.

    If $\lambda_1 < -1$ and $\lambda_2 = -1$, we have,
    $$
    g'(x; \lambda_1, -1) = -\frac{\lambda_1^2x^{-2}(x^{\lambda_1 + 1} - (\lambda_1^{-1} + 1)x^{\lambda_1}+\lambda_1^{-1})}{(x^{\lambda_1}-1)^2} = -\frac{(1+\lambda_1)x^{-2}\phi_{\lambda_1^{-1}}(x^{\lambda_1})}{(x^{\lambda_1}-1)^2}.
    $$
    The denominator is obviously positive for $x\in(0,1)$ and $\lambda_1 < -1$. The term $\phi_{\lambda_1^{-1}}(x^{\lambda_1})$ in the numerator is non-negative for all $\lambda_1 < -1$ and $x\geq 0$ by definition, and $-(1+\lambda_1)x^{-2} > 0$ for $x\in(0,1)$ and $\lambda_1 < -1$. It then follows that $g'(x; \lambda_1, -1) \geq 0$ for $x\in(0,1)$. This then shows that $g(x; \lambda_1, -1)$ is non-decreasing over $x\in(0,1)$.

    If $\lambda_1 < 0$ and $\lambda_2 = 0$, there are two cases. We have already dealt with the case where $\lambda_1 = -1$ above. Otherwise, for $\lambda_1 \neq -1$, we have,
    $$
    g'(x;\lambda_1; 0) = -\frac{\lambda_1x^{-1}(1-x^{\lambda_1} +x^{\lambda_1} \log(x^{\lambda_1}))}{(x^{\lambda_1} - 1)^2} = -\frac{\lambda_1x^{-1}\phi_0(x^{\lambda_1})}{(x^{\lambda_1}-1)^2}.
    $$
    The denominator is obviously positive for $x \in (0, 1)$ and $\lambda_1 < 0$. The numerator is non-negative since $\phi_0(x^{\lambda_1}) \geq 0$ for all $\lambda_1 < 0$ and $x \geq 0$ by definition, and $-\lambda_{1} x^{-1} > 0$ for $x\in(0,1)$ and $\lambda_1 < 0$. It then follows that $g'(x; \lambda_1, 0)$ is non-negative. Therefore, $g(x; \lambda_1, 0)$ is non-decreasing over $x\in(0,1)$.

    If $-\infty < \lambda_1 < \lambda_2 < \infty$ and $\lambda_1, \lambda_2 \neq -1, 0$, write $\lambda_2 = \lambda_1 + \kappa$ for some positive constant $\kappa > 0$. Then, we have,
    \begin{align*}
        g'(x; \lambda_1, \lambda_1+\kappa) &= -\frac{\lambda_1x^{\lambda_1-1}(x^\kappa -(1 - (\lambda_1+\kappa)^{-1}\lambda_1)x^{\lambda_1 + \kappa} - (\lambda_1 + \kappa)^{-1}\lambda_1)}{(x^{\lambda_1}-1)^2}\\ &= \frac{\lambda_1x^{\lambda_1 - 1}(1-(\lambda_1 + \kappa)^{-1}\lambda_1)(\lambda_1 + \kappa)^{-1}\lambda_1\phi_{-(\lambda_1 + \kappa)^{-1}\lambda_1}(x^{\lambda_1 + \kappa})}{(x^{\lambda_1} - 1)^2}
        \end{align*}
    In terms of $\lambda_1$ and $\lambda_2$, this can be written as,
    \begin{align*}
       g'(x; \lambda_1, \lambda_2) &= \frac{\lambda_1x^{\lambda_1 - 1}[(1-\lambda_2^{-1}\lambda_1)\lambda_2^{-1}\lambda_1]\phi_{-\lambda_2^{-1}\lambda_1}(x^{\lambda_2})}{(x^{\lambda_1} - 1)^2}.
    \end{align*}
    It is obvious that the denominator is positive for $x\in(0,1)$ and $\lambda_1 \neq -1, 0$. As for the numerator, $x^{\lambda_1 - 1}$ is always non-negative for $x\in(0,1)$ and $\lambda_1 \neq -1, 0$. The term $\phi_{-\lambda_2^{-1}\lambda_1}(x^{\lambda_2}) \geq 0$ for $x \geq 0$ for all $\lambda_1, \lambda_2 \neq -1, 0$ by definition. The remaining terms in the numerator can be expressed as,
    $$
    \lambda_1 \left(1 - \frac{\lambda_1}{\lambda_2}\right) \left(\frac{\lambda_1}{\lambda_2}\right) = \left(\lambda_2 - \lambda_1\right) \left(\frac{\lambda_1}{\lambda_2}\right)^2. 
    $$
    Since $\lambda_2 > \lambda_1$ by definition, this term is always positive for all $\lambda_1, \lambda_2 \neq -1, 0$. Hence, we have shown that all terms in the numerator and denominator are either non-negative or positive, so it follows that $g'(x; \lambda_1, \lambda_2)$ is also non-negative for all $\lambda_1, \lambda_2 \neq -1, 0$.

    Considering all cases together, we see that $g(x; \lambda_1, \lambda_2)$ is a non-decreasing function of $x \in (0,1)$ for all $-\infty < \lambda_1 < \lambda_2 < \infty$, and the theorem follows.  
\end{proof}

\begin{proof}[Proof of Proposition \ref{prop:FHLB_and_FHUB}]
    For part (i), from Theorem 4.4.7 in \citet{Nelsen2006}, it is sufficient to show that 
    $$
    \lim_{\lambda \to \infty} \frac{\phi_\lambda(s)}{\phi'_\lambda(t)} = s - 1, ~s, t \in (0, 1).
    $$
    With appropriate substitutions and an application of L'H\^opital's rule, we obtain, 
    $$
    \lim_{\lambda \to \infty} \frac{s^{\lambda + 1} - s + \lambda(1-s)}{(\lambda + 1)(t^\lambda - 1)} = \lim_{\lambda \to \infty} \frac{s^{\lambda + 1}\log(s) + 1 - s}{t^\lambda - 1 + (\lambda + 1)t^\lambda\log(t)} = s-1,
    $$
    where recall $s, t \in (0, 1)$. Hence part (i) of the proposition follows.

    For part (ii), from Theorem 4.4.8 in \citet{Nelsen2006}, it is sufficient to show that 
    $$
    \lim_{\lambda \to -\infty} \frac{\phi_\lambda(s)}{\phi'_\lambda(s)} = 0, ~s \in (0, 1).
    $$
    With appropriate substitutions and two applications of L'H\^opital's rule, we obtain,
    $$
    \lim_{\lambda \to -\infty} \frac{s^{\lambda + 1} - s + \lambda(1-s)}{(\lambda + 1)(s^\lambda - 1)} = \lim_{\lambda \to -\infty} \frac{s^{\lambda + 1}\log(s)^2}{(\lambda + 1)s^\lambda\log(s)^2 + 2s^\lambda \log(s)},~s\in (0,1).
    $$
    The limit on the right-hand side (RHS) is $0$ as $\lambda \to -\infty$. Hence part (ii) of the proposition follows. 
\end{proof}

\begin{proof}[Proof of Proposition \ref{prop:Kendalls_tau_lambda}]
    Let $-\infty < \lambda_1 < \lambda_2 < \infty$ be real-valued parameters. Kendall's tau obeys the ordering of the copula family \citep[e.g.,][Def 5.1.7]{Nelsen2006}. Due to Theorem \ref{thm:ordering}, $C_{\lambda_1}(u_1, u_2) \geq C_{\lambda_2}(u_1, u_2)$ for all $u_1, u_2 \in [0, 1]$. Therefore, recalling $\tau(\lambda)$ in \eqref{eqn:Kendalls_tau}, we have $\tau(\lambda_1) \geq \tau(\lambda_2)$, and $\tau(\lambda)$ must be a non-increasing function of $\lambda \in (-\infty, \infty)$.
\end{proof}

\begin{proof}[Proof of Theorem \ref{thm:tail_dependence}]
For part (i), it is obvious that the lower-tail dependence coefficient is zero for $\lambda > -1$ because, due to Theorem \ref{thm:zero_set}, the copula assigns zero probability to any values of $\U \equiv (U_1, U_2)^\top$ below the zero curve; hence $\lim_{u \downarrow 0} \Pr(U_1 \leq u\mid U_2 \leq u) = 0$.  

When $\lambda \leq -1$, we apply the results of \citet{Charpentier2009} to calculate the lower tail-dependence coefficient. In what follows, we write $\lambda = - \gamma$, where $\gamma \geq 1$, for ease of exposition. Now, we must calculate,
$$
\nu_0(-\gamma) = - \left\{\lim_{s \downarrow 0} \frac{s\phi'_{-\gamma}(s)}{\phi_{-\gamma}(s)}\right\},~\gamma \geq 1.
$$
First, consider the case where $\gamma = 1$ (i.e., $\lambda = -1$). Here, 
$$
\nu_0(-1) = -\left\{\lim_{s\downarrow 0} \frac{s-1}{s - 1 - \log(s)}\right\} = 0,
$$
by inspection. Now for the $\gamma > 1$ case (i.e., $\lambda > -1$), define,
$$
\nu_0(-\gamma) = - (1-\gamma) \left\{\lim_{s\downarrow 0} \frac{s^{1-\gamma} - s}{s^{1-\gamma}-s-\gamma(1-s)}\right\},~\gamma > 1.
$$
After two applications of L'H\^opital's rule, this can be evaluated as,
$$
\nu_0(-\gamma) = -(1 - \gamma)\left\{\lim_{s\downarrow 0}\frac{(1-\gamma)\gamma s^{-(\gamma+1)}}{(1-\gamma)\gamma s^{-(\gamma+1)}}\right\} = -(1-\gamma), ~\gamma>1.
$$
Then, combining both cases, we obtain,
$$
\nu_0(-\gamma) = \begin{cases}
    -(1 - \gamma) & \gamma \neq 1,\\
    0 & \gamma = 1. 
\end{cases}
$$
Then, from Theorem 3.1 of \citet{Charpentier2009}, if $\nu_0(-\gamma) = 0$, the lower-tail dependence coefficient is zero. If $\nu_0 > 0$, the lower-tail dependence coefficient is $T_L(-\gamma) = 2^{-1/\nu_0(-\gamma)} = 2^{1/(1-\gamma)}$. Therefore, in terms of $\lambda$, we can express the lower tail-dependence coefficient of the PD copulas as,
$$
T_L(\lambda) = \begin{cases}
   2^{1/(\lambda+1)} & \lambda \neq -1,\\
   0 & \lambda \geq -1,
\end{cases}
$$
as required.

For part (ii), from \citet{Charpentier2009}, the upper tail-dependence properties of the PD copula can be discovered by calculating,
    $$
    \nu_1(\lambda) = - \left\{\lim_{s\downarrow 0} \frac{s\phi_\lambda'(1-s)}{\phi_\lambda(1-s)}\right\}.
    $$
    We must consider all $\lambda \in (-\infty, \infty)$, so we calculate this quantity for three cases ($\lambda = -1$, $\lambda = 0$, and $\lambda \neq -1, 0$). First, for $\lambda = -1$, 
    $$
    \nu_1(-1) = \lim_{s\downarrow0} \frac{s[1-(1-s)^{-1}]}{s + \log(1-s)}.
    $$
    This evaluates to $\nu_1(-1) = 2$ after two applications of L'H\^opital's rule. Second, for $\lambda = 0$, we have,
    $$
    \nu_1(0) = -\left\{\lim_{s \downarrow 0} \frac{s\log(1-s)}{s + (1-s)\log(1-s)}\right\}.
    $$
    This evaluates to $\nu_1(0)=2$ with two applications of L'H\^opital's rule. Finally, when $\lambda \neq -1, 0$, we have,
    $$
    \nu_1(\lambda) = -(\lambda+1)\left\{\lim_{s \downarrow 0} \frac{s[(1-s)^\lambda - 1]}{(1-s)^{\lambda + 1} - (\lambda+1)(1-s) + \lambda}\right\},~\lambda \neq -1, 0.
    $$
    Two applications of L'H\^opital's rule gives,
    $$
    \nu_1(\lambda) = -(\lambda+1)\left\{\lim_{s \downarrow 0} \frac{\lambda(\lambda-1)s(1-s)^{\lambda-2}-2\lambda(1-s)^{\lambda - 1}}{\lambda(\lambda+1)(1-s)^{\lambda-1}}\right\}=2.
    $$
    Finally, from Theorem 4.1 of \citet{Charpentier2009}, since $\nu_1(\lambda) > 0$, the upper tail-dependence coefficient can be expressed as a function of $\nu_1(\lambda)$, namely $T_U(\lambda) = 2 - 2^{1/\nu_1(\lambda)}$. For all $\lambda \in (-\infty,\infty)$, we have seen that $\nu_1(\lambda) = 2$. Therefore, for $C_\lambda(u_1, u_2)$, $T_U(\lambda) = 2 - \sqrt{2}$ for all $\lambda \in (-\infty,\infty)$, as required. 
\end{proof}

\begin{proof}[Proof of Lemma \ref{lem:no_derivative_at_phi0}]
    Let $\lambda > 0$. Recall that,
    \begin{equation}
    \phi_\lambda^{[-1]}(t) = \begin{cases}
        \phi_\lambda^{-1}(t) & 0 \leq t < 1/(\lambda + 1),\\
        0 & 1/(\lambda + 1) \leq t < \infty,
    \end{cases}\label{eqn:appendix_pseudoinverse}
    \end{equation}
    where $\phi^{-1}_\lambda(t)$ is given by the appropriate solution of \eqref{eqn:zeros_of_pseudopolynomial}. The function \eqref{eqn:appendix_pseudoinverse} is continuous over $t \in [0,\infty)$.

    Recall that $\phi_\lambda(x)$ in \eqref{eqn:CR_phi_function} with  $\lambda > 0$ has at least one derivative over $x \in (0, \infty)$, which is given by $\phi_\lambda'(x)$ in \eqref{eqn:CR_derivative}. Therefore, the Inverse Function Theorem (IFT) implies that $(\phi_\lambda^{-1})'(t)$ exists over $t \in (0, 1/(\lambda + 1))$ and is equal to, 
    $$
    (\phi_\lambda^{-1})'(t) = \frac{1}{\phi_\lambda'(\phi^{-1}_{\lambda}(t))} = \frac{\lambda}{\phi_\lambda^{-1}(t)^\lambda - 1}.
    $$ 
    Then, over $t \in (0, 1/(\lambda + 1))$, we have $(\phi_\lambda^{[-1]})'(t) = (\phi_\lambda^{-1})'(t)$. Similarly, over $t \in (1/(\lambda + 1), \infty)$, we have $(\phi_\lambda^{[-1]})'(t) = 0$. 
    
    Now consider the one-sided limit of $(\phi_\lambda^{[-1]})'(t)$ as $t \to 1/(\lambda+1)$ from below, denoted as $t \uparrow 1/(\lambda + 1)$. This is, 
    $$
    \lim_{t\uparrow 1/(\lambda + 1)}(\phi_\lambda^{[-1]})'(t) = \lim_{t\uparrow 1/(\lambda + 1)}(\phi_\lambda^{-1})'(t) = \lim_{t\uparrow 1/(\lambda+1)} \frac{\lambda}{\phi_\lambda^{-1}(t)^\lambda - 1} = -\lambda.
    $$
    Also take the other one-sided limit as $t$ approaches $1/(\lambda + 1)$ from above (i.e., $t \downarrow 1/(\lambda + 1)$): That is,
    $$
    \lim_{t\downarrow 1/(\lambda + 1)}(\phi_\lambda^{[-1]})'(t) = 0.
    $$
    The limit of the left-derivative and right-derivative do not agree at $t = 1/(\lambda + 1)$, which implies the derivative of \eqref{eqn:appendix_pseudoinverse} does not exist there. 
\end{proof}

\begin{proof}[Proof of Theorem \ref{thm:no_pd_copulas_d_geq_3_lambda_above_0}]
    The existence of the derivative of $\phi^{[-1]}_\lambda(t)$ at all $t \in (0, \infty)$ is a necessary condition for $\phi^{[-1]}_\lambda$ to be $d$-monotone for any $d \geq 3$, which is in turn  sufficient and necessary for $\phi_{\lambda}^{[-1]}(\phi_\lambda(u_1) + \cdots + \phi_\lambda(u_d))$ to be a valid copula \citep{McNeil2009}. Lemma \ref{lem:no_derivative_at_phi0} establishes that $\phi_\lambda^{[-1]}(t)$ lacks a derivative at $t = 1/(\lambda + 1)$ when $\lambda > 0$, so it cannot even be $3$-monotone. Hence the theorem follows.
\end{proof}

\begin{proof}[Proof of Lemma \ref{lem:derivative_exists}]
    Let $\lambda = 0$. The pseudoinverse $\phi_{0}^{[-1]}(t)$ is given by \eqref{eqn:pseudoinverse_0} for $\lambda = 0$. Then, when $t \in (0, 1)$, $(\phi_0^{[-1]})'(t) = (\phi_0^{-1})'(t)$; otherwise, when $t \in (1, \infty)$, $(\phi_0^{[-1]})'(t) = 0$. The derivative $(\phi_0^{-1})'(t)$ exists and, from \eqref{eqn:pseudoinverse_0}, Mathematica \citep{Mathematica} gives the following result: For $t \in (0, 1)$, 
    \begin{equation}
        (\phi_{0}^{-1})'(t) = \frac{1}{1 + \gW_{-1}((t-1)/\exp\{1\})},\label{eqn:first_deriv_zero}
    \end{equation}
    where recall that $\gW_{-1}$ is the lower branch of the Lambert W function. Then, by the properties of the lower branch of the Lambert W function, the limit of $(\phi_0^{[-1]})'(t)$ as $t \uparrow 1$ is,
    $$
    \lim_{t\uparrow 1} (\phi_0^{[-1]})'(t) = \lim_{t\uparrow 1} (\phi_0^{-1})'(t) = \lim_{t\uparrow 1} \frac{1}{1 + \gW_{-1}((t-1)/\exp\{1\})} = 0. 
    $$
    The other one-sided limit as $t \downarrow 1$ agrees, since $\lim_{t\downarrow 1} (\phi_0^{[-1]})'(t) = 0$ by inspection. Therefore, $(\phi_0^{[-1]})'(t)$ exists at $t = \phi_0(0) = 1$. 
   
    Now let $-1 < \lambda < 0$, and write $\lambda = -\gamma$ with $0 < \gamma < 1$. The pseudoinverse $\phi_{-\gamma}^{[-1]}(t)$ is given by \eqref{eqn:appendix_pseudoinverse} for $-1 < \lambda < 0$ (i.e., $0 < \gamma < 1$). For $t \in (1/(1-\gamma), \infty)$, $(\phi^{[-1]}_{-\gamma})'(t) = 0$. For $t \in (0, 1/(1- \gamma))$, $(\phi^{[-1]}_{-\gamma})'(t) = (\phi^{-1}_{-\gamma})'(t)$. Since $\phi_{-\gamma}(x)$ has infinitely many derivatives on $x \in (0, \infty)$, the IFT implies that $(\phi_{-\gamma}^{-1})'(t)$ exists for $t \in (0, 1/(1-\gamma))$ and is given by,
    \begin{equation}
        (\phi^{-1}_{-\gamma})'(t) =\frac{1}{\phi_{-\gamma}'(\phi_{-\gamma}^{-1}(t))}= -\frac{\gamma}{\phi_{-\gamma}^{-1}(t)^{-\gamma} - 1} = -\frac{\gamma \phi_{-\gamma}^{-1}(t)^\gamma}{1 - \phi_{-\gamma}^{-1}(t)^\gamma}.\label{eqn:first_deriv}
    \end{equation}
    Now, evaluate the one-sided limit of $(\phi^{[-1]}_{-\gamma})'(t)$ as $t \uparrow 1/(1-\gamma)$: This is,
    $$
    \lim_{t\uparrow 1/(1-\gamma)} (\phi^{[-1]}_{-\gamma})'(t) = \lim_{t\uparrow 1/(1-\gamma)} (\phi^{-1}_{-\gamma})'(t) = -\left\{\lim_{t\uparrow 1/(1-\gamma)} \frac{\gamma \phi_{-\gamma}^{-1}(t)^\gamma}{1 - \phi_{-\gamma}^{-1}(t)^\gamma}\right\} = 0, 
    $$
    since $0 \leq \phi_{-\gamma}^{-1}(t) \leq 1$ and $0 < \gamma < 1$. The other one-sided limit is $\lim_{t\downarrow 1/(1-\gamma)}(\phi^{[-1]}_{-\gamma})'(t) = 0$. These two limits agree, so $(\phi^{[-1]}_{-\gamma})'(t)$ exists at $t = 1/(1-\gamma)$ (and at all $t \in (0, \infty)$. 

    Combining the two cases above ($\lambda = 0$ and $-1 < \lambda < 0$) gives the result.
\end{proof}

\begin{proof}[Proof of Lemma \ref{lem:fails_to_be_convex}]
    Let $\lambda = 0$. Recall that $(\phi_0^{-1})'(t)$ for $t \in (0, 1)$ is given by \eqref{eqn:first_deriv_zero}. Now, for the second derivative, Mathematica gives,
    \begin{align}
        (\phi_{0}^{-1})^{''}(t) &= -\frac{\gW_{-1}((t-1)/\exp\{1\})}{(t-1)(1 + \gW_{-1}((t-1)/\exp\{1\}))^3} \nonumber\\
        &= \frac{\gW_{-1}((t-1)/\exp\{1\})}{(1-t)(1 + \gW_{-1}((t-1)/\exp\{1\}))^3},\label{eqn:second_derivative_zero}
    \end{align}
    where recall $t \in (0, 1)$. We see that \eqref{eqn:second_derivative_zero} is positive over all $t \in (0, 1)$ since $(1 - t) > 0$ and $-\infty < \gW_{-1}(s) < -1$ for all $s \in (-\exp\{-1\}, 0)$. Since \eqref{eqn:second_derivative_zero} is positive over its domain, the first derivative of $-(\phi_{0}^{-1})'(t)$, which is simply $-1$ times \eqref{eqn:second_derivative_zero}, is negative over its domain, demonstrating that $-(\phi_{0}^{-1})'(t)$ is non-increasing, as required. As for convexity, simply plotting $-(\phi_0^{-1})'(t)$ over $t \in (0, 1)$ reveals that this function is not convex on $t \in (0, 1)$; therefore, $-(\phi_{0}^{[-1]})'(t)$ cannot be convex on $t \in [0, \infty)$.
    
    Now let $-1 < \lambda < 0$, and write $\lambda = -\gamma$ for $0 < \gamma < 1$. Then recall that $(\phi_{-\gamma}^{-1})'(t)$ is given by \eqref{eqn:first_deriv}. On $t \in (0, 1/(1-\gamma))$, the second and third derivatives exist and are given by,
    \begin{align}
        (\phi_{-\gamma}^{-1})^{''}(t) &= \frac{\gamma^3\phi_{-\gamma}^{-1}(t)^{2\gamma - 1}}{(1-\phi_{-\gamma}^{-1}(t)^\gamma)^3}.\label{eqn:second_derivative}\\        
        (\phi_{-\gamma}^{-1})^{(3)}(t) &= -\left\{\frac{\gamma^4 \phi_{-\gamma}^{-1}(t)^{3\gamma-2}((1+\gamma)\phi_{-\gamma}^{-1}(t)^{\gamma} + 2\gamma - 1)}{(1-\phi_{-\gamma}^{-1}(t)^{\gamma})^5}\right\}\label{eqn:third_derivative}
    \end{align}
    Observe that \eqref{eqn:second_derivative} is non-negative for all $t\in(0, \infty)$ since $0 \leq \phi_{-\gamma}^{-1}(t) \leq 1$. Then, $-(\phi^{-1}_{-\gamma})'(t)$ must have non-positive first derivative given by $-1$ times \eqref{eqn:second_derivative}, so $-(\phi^{-1}_{-\gamma})'(t)$ is non-increasing, as required, for $0 < \gamma < 1$ (equivalently, $-1 < \lambda < 0$). As for convexity, test it by using the second derivative of $-(\phi^{-1}_{-\gamma})'(t)$, which is simply $-1$ times \eqref{eqn:third_derivative}. Therefore, the function $-(\phi^{-1}_{-\gamma})'(t)$ is convex on $t\in(0,1/(1-\gamma))$ if and only if,
    $$
    \frac{\gamma^4 \phi_{-\gamma}^{-1}(t)^{3\gamma-2}((1+\gamma)\phi_{-\gamma}^{-1}(t)^{\gamma} + 2\gamma - 1)}{(1-\phi_{-\gamma}^{-1}(t)^{\gamma})^5} \geq 0,
    $$
    for all $t \in (0, 1/(1-\gamma))$. Note that the denominator on the left-hand side (LHS) of the inequality is always non-negative since $0 \leq \phi_{-\gamma}^{-1}(t) \leq 1$. But, in the numerator, the term $((1+\gamma)\phi_{-\gamma}^{-1}(t)^{\gamma} + 2\gamma - 1)$ will be negative for some $t \in (0, 1/(1-\gamma))$ if $\gamma < 0.5$ (i.e., $\lambda > -0.5$). If $\gamma \geq 0.5$ (i.e., $\lambda \leq -0.5$), then non-negativity is guaranteed. Therefore, $-(\phi^{-1}_{-\gamma})'(t)$ is convex when $0.5 \leq \gamma < 1$ (equivalently, $-1 < \lambda \leq -0.5$), but not when $0 < \gamma < 0.5$ (equivalently, $-0.5 < \lambda < 0$). It then follows that $-(\phi^{[-1]}_\lambda)'(t)$ is convex when $-1 < \lambda \leq -0.5$ but not when $-0.5 < \lambda < 0$.

    Combining the two cases gives the result.
\end{proof}

\begin{proof}[Proof of Theorem \ref{thm:three_monotone}]
    This follows immediately from Lemmas \ref{lem:derivative_exists} and \ref{lem:fails_to_be_convex} and Proposition 2.3 of \citet{McNeil2009}. 
\end{proof}

\begin{proof}[Proof of Lemma \ref{lem:pd_copulas_completely_monotone}]
    We first show for $\lambda = -1$. To show that $\phi_{-1}^{-1}$ in \eqref{eqn:pseudoinverse_neg1} is completely monotone over $t\in (0,\infty)$, rewrite \eqref{eqn:pseudoinverse_neg1} as $\phi_{-1}^{-1}(t) = \gV(\exp\{-{(t+1)}\})$, where $\gV(r) \equiv -\gW_0(-r)$ for $r \in (0, \exp\{-1\})$. It is obvious that $\exp\{-(t+1)\}$ is completely monotone on $t \in (0, \infty)$ since the function is positive, and the chain rule shows its $k$-th derivative is $(-1)^k\exp\{-(t+1)\}$ for $k = 1, 2, ...$. As for $\gV(r)$, this function is absolutely monotone on $r\in(0, \exp\{-1\})$, as shown below. The function itself is non-negative, from the properties of the principal branch of the Lambert W function. The first derivative is $\gV'(r) = \gW_0'(-r)$. It has previously been shown \citep{Kalugin2012} that $\gW_0'(r)$ is completely monotone on $r\in (-\exp\{-1\}, \infty)$, so it is also completely monotone on the subinterval $r \in (0, \exp\{-1\})$. Therefore, $\gW_0'(-r)$ is absolutely monotone on $r \in (0,\exp\{-1\})$ \citep[p. 145]{Widder1946}. Therefore, being non-negative and having absolutely monotone first derivative, it follows that $\gV(r)$ is absolutely monotone on $r\in (0, \exp\{-1\})$. Now Theorem 2b of \citet[p. 145]{Widder1946} implies $\phi_{-1}^{-1}(t) = -\gW_0(-\exp\{-(t+1)\})$ is completely monotone on $t \in (0,\infty)$ since this is a composite of an absolutely monotone function and a completely monotone function.

 Next we show for $\lambda = -2$. Let $z = 1$, and set $\lambda = -(z+1) = -2$. We first establish that $\phi^{-1}_{-2}$ has derivatives of all orders over $t \in (0, \infty)$. Then we show that the $k$-th derivative of $\phi^{-1}_{-2}$ satisfies $(-1)^{k}(\phi^{-1}_{-2})^{(k)}(t) \geq 0$ for all $k = 1, 2, ...$. 
    
    For the existence of the derivatives of all orders, when $\lambda = -2$, \eqref{eqn:CR_phi_function} has the form, $\phi_{-2}(x) = 0.5(x^{-1} + x - 2)$. The derivatives are $\phi^{(1)}_{-2}(x) = -0.5(x^{-2} - 1)$, $\phi^{(2)}_{-2}(x) = x^{-3}$, and $
    \phi^{(k)}_{-2}(x) = (-1)^{k} \times \left(\prod_{j=3}^kj\right) \times x^{-(k+1)}$ for $k = 3, 4, ...$. Hence, $\phi_{-2}(x)$ is infinitely differentiable over $x\in(0, \infty)$ for $\lambda = -2$. Then, the IFT implies $\phi_{-2}^{-1}(t)$ has derivatives of all orders over $t\in (0,\infty)$ since $\phi_{-2}(x)$ is continuous and has derivatives of all orders for all $x\in (0, \infty)$. 
    
    Next, for the signs of the derivatives, rearrange \eqref{eqn:zeros_of_pseudopolynomial} into the form,
    \begin{equation}
    0 = x(t)^{2} - 2[t + 1]x(t)+1,   \label{eqn:appendix_polynomial} 
    \end{equation}
    where $x(t) \in [0,1]$ is the function of $t \in [0,\infty)$ that solves this equation. Hence, $x(t)$ is precisely $\phi_{-2}^{-1}(t)$. 
    
    Now let $k \geq 1$ be a positive integer, and differentiate both sides of \eqref{eqn:appendix_polynomial} $(k+1)$ times with respect to $t$. On the RHS of \eqref{eqn:appendix_polynomial}, the $(k+1)$-th derivative of $x(t)^{2}$ is given by Fa\`{a} di Bruno's formula, and the $(k+1)$-th derivative of $2[t + 1]x(t)$ can be obtained by applying the general Leibniz rule. This yields,
    \begin{align*}
    0 = 2x(t)\gB_{k+1,1}(t) + 2\gB_{k+1,2}(t) - 2(k+1)x^{(k)}(t) - 2[t+1]x^{(k+1)}(t),
    \end{align*}
    where $\gB_{k+1,1}(t) \equiv B_{k+1,1}(x^{(1)}(t), ..., x^{(k+1)}(t))$ and $\gB_{k+1,2}(t) \equiv B_{k+1,2}(x^{(1)}(t), ..., x^{(k)}(t))$ are incomplete exponential Bell polynomials per Definition \ref{def:Bell_polynomial} \citep[e.g.,][pp. 133-137]{Comtet1974}. (Here, we have written the Bell polynomials as functions of $t$ since the arguments to the Bell polynomial are functions of $t$.) Importantly, the definition in \eqref{eqn:Bell_polynomial} and conditions in \eqref{eqn:sum_to_j} and \eqref{eqn:total_weight_sum_to_k} imply that $\gB_{k+1,1}(t) = x^{(k+1)}(t)$. Hence, after some algebra, we obtain,
    \begin{align}
        0 = \left(x(t) - [t+1]\right)x^{(k+1)}(t) + \gB_{k+1,2}(t) - (k+1)x^{(k)}(t).\label{eqn:appendix_after_simplication}
    \end{align}
    By rearranging \eqref{eqn:appendix_after_simplication} to isolate $x^{(k+1)}(t)$ on the LHS, we write,
    \begin{align}
        x^{(k+1)}(t) = -\left(\frac{(k+1)x^{(k)}(t) - \gB_{k+1,2}(t)}{t+1 - x(t)}\right).\label{eqn:k_plus_oneth_derivative}
    \end{align}
    
    From this point, the proof is completed by induction on $k$. The inductive hypothesis is that, for given positive integer $k$, the derivatives $x^{(j)}(t)$ all satisfy $(-1)^jx^{(j)}(t) \geq 0$ for $j = 1, ..., k$. With this inductive hypothesis in mind, examine the terms on the RHS of \eqref{eqn:k_plus_oneth_derivative}. First, there is a prefactor of $-1$. Second, notice that the denominator $t + 1 - x(t) \geq 0$ since $t \in [0, \infty)$ and $x(t) \in [0, 1]$. Third, in the numerator, we have $(k+1)x^{(k)}(t) - \gB_{k+1,2}(t)$. See that $(k+1) > 0$ by definition. By the inductive hypothesis, $x^{(k)}(t)$ has sign $(-1)^{k}$ (i.e., positive if $k$ is even and negative if $k$ is odd) and, also by Lemma \ref{lem:Bell_polynomials}, we have that $\gB_{k+1,2}(t) \geq 0$ if $k+1$ is even ($k$ is odd) and $\gB_{k+1,2}(t) \leq 0$ if $k+1$ is odd ($k$ is even). Therefore, $(k+1)x^{(k)}(t)$ and $\gB_{k+1,2}(t)$ have opposite signs, and $(k+1)x^{(k)}(t) - \gB_{k+1,2}(t)$ has the same sign as $x^{(k)}(t)$, namely $(-1)^k$. Hence $x^{(k+1)}(t)$ has the opposite sign to $x^{(k)}(t)$ for any given $k = 1, 2, ...$. In other words, if $(-1)^k x^{(k)}(t) \geq 0$, then $(-1)^{k+1}x^{(k+1)}(t) \geq 0$.  

    The base case for the induction is provided by explicitly calculating $x^{(1)}(t)$ and verifying that $(-1)x^{(1)}(t) \geq 0$. It is easier to verify this by working directly from $\phi^{-1}_{-2}(t)$ than by using \eqref{eqn:k_plus_oneth_derivative}. By the IFT, $x(t) \equiv \phi^{-1}_{-2}(t)$ has first derivative,
    $$
    x^{(1)}(t) = \frac{1}{\phi_{-2}'(x(t))}=-\left\{\frac{2}{x(t)^{-2}-1}\right\} = -2\left\{\frac{x(t)^{2}}{1-x(t)^{2}}\right\},
    $$
    which is negative since $0 \leq x(t)^{2} \leq 1$, which in turn indicates that $(-1)x^{(1)}(t) \geq 0$. This provides the base case. Therefore, by induction, it follows that $(-1)^kx^{(k)}(t) \geq 0$ for all $k = 1, 2,  ...$. Hence $x(t)\equiv\phi_{-2}^{-1}(t)$ is completely monotone on $t\in[0,\infty)$. 
\end{proof}

\begin{proof}[Proof of Theorem \ref{thm:completely_monotone_cases}]
    This follows immediately from Lemma \ref{lem:pd_copulas_completely_monotone} in combination with the result of \citet{Kimberling1974}; see also Theorem 4.6.2 in \citet{Nelsen2006}.
\end{proof}

\end{document}